\documentclass{fundam}

\usepackage{url} 
\usepackage{graphicx}

\usepackage[english]{babel}
\usepackage[utf8]{inputenc}
\usepackage{stmaryrd}
\usepackage{algorithm}
\usepackage{algorithmic}

\usepackage{xspace}
\usepackage{amsmath,amssymb}
\usepackage{dsfont}
\usepackage{mathtools}
\usepackage{tikz}
\usepackage{subfigure}

\usetikzlibrary{petri,positioning}
\tikzstyle{normal place}=[circle,thick, draw=black!100,fill=white!100,minimum size=6mm]
\tikzstyle{gray place}=[circle,thick, draw=black!100,fill=black!35,minimum size=6mm]
\tikzstyle{grayclair place}=[circle,thick, draw=black!100,fill=black!5,minimum size=6mm]
\tikzstyle{contour place}=[place,draw=red!100]
\tikzstyle{purple place}=[place,draw=purple!100,fill=red!20]
\tikzstyle{orange place}=[place,draw=black!100,fill=orange!20]
\tikzstyle{brown place}=[place,draw=black!100,fill=brown!70]
\tikzstyle{brownclair place}=[place,draw=black!100,fill=brown!30]
\tikzstyle{blue place}=[place,draw=blue!100,fill=blue!30]
\tikzstyle{normal transition}=[ rectangle,thick, draw=black!100,fill=black!5, minimum width=8mm, inner ysep=2pt]	
\tikzstyle{purple transition}=[ rectangle,thick, draw=purple!100,fill=purple!35, minimum width=8mm, inner ysep=2pt]	
\tikzstyle{yellow transition}=[ rectangle,thick, draw=black!100, fill=yellow!20, minimum width=8mm, inner ysep=2pt]	
\tikzstyle{green transition}=[ rectangle,thick, draw=black!100, fill=green!20, minimum width=8mm, inner ysep=2pt]	
\tikzstyle{blue transition}=[ rectangle,thick, draw=blue!100,fill=blue!25, minimum width=8mm, inner ysep=2pt]	

\newcommand{\en}[1]{\mathsf{en}(#1)}
\newcommand{\firable}[1]{\mathsf{firable}(#1)}
\newcommand{\newen}[1]{\mathsf{newen}(#1)}

\newcommand{\calN}{{\cal N}}

\newcommand{\last}{{\sf last}}
\newcommand{\lastm}{{\sf lastm}}
\newcommand{\Runs}{{\sf Runs}}
\newcommand{\N}[0]{\mathbb{N}}
\newcommand{\Z}[0]{\mathbb{Z}}
\newcommand{\Q}[0]{\mathbb{Q}}
\newcommand{\R}[0]{\mathbb{R}}
\newcommand{\Rp}{\mathbb{R}_{\geq 0}}
\newcommand{\Qp}{\mathbb{Q}_{\geq 0}}
\newcommand{\pre}[1]{\ensuremath{{}^{\bullet}\mspace{-2mu}#1}}
\newcommand{\post}[1]{\ensuremath{#1^{\bullet}}}
\newcommand{\Markings}[0]{{\N}^{P}}

\newcommand{\states}[0]{\Markings \times \Intervals{\Rp}^{T}}
\newcommand{\Intervals}[1]{\mathcal{I}(#1)}
\newcommand{\Net}[0]{\mathcal{N}}
\newcommand{\minusI}{\ominus}
\newcommand{\eft}[1]{\underline{#1}}
\newcommand{\lft}[1]{\overline{#1}}
\newcommand{\proj}[2]{{#1}_{|#2}}
\newcommand{\Next}{{\sf Next}}
\newcommand{\cost}[0]{\textsf{cost}}
\newcommand{\costrate}[0]{\cost_m}
\newcommand{\discreteCost}[0]{\cost_t}
\newcommand{\intfirable}{$\N^\Params$-firable}
\newcommand{\costclass}[0]{\cost}
\newcommand{\costclassint}[0]{\costclass_\N}
\newcommand{\costrun}[0]{\cost}

\newcommand{\Params}{\mathbb{P}}
\newcommand{\sequence}{\mathsf{sequence}}

\def\Conv{{\sf Conv}}
\def\IV{{\sf Ints}}
\def\IH{{\sf IH}}

\newcommand{\aCost}{\textsc{Cost}}
\newcommand{\aPassed}{\textsc{Passed}}
\newcommand{\aWaiting}{\textsc{Waiting}}
\newcommand{\aGoal}{\textsf{Goal}}
\newcommand{\aPoly}{\textsf{PolyRes}}

\newcommand{\costincl}{\preccurlyeq}
\newcommand{\costrincl}{\succcurlyeq}
\newcommand{\cmax}{c_{\max}}

\usepackage{tikz}
\usetikzlibrary{automata, arrows, positioning}

\tikzstyle{location}=[minimum size=12pt, circle, fill=blue!20, draw=black, text=black, inner sep=1.5pt, ]

    \tikzstyle{place}=[circle,thick,draw=black!75,fill=magenta!20,minimum size=6mm]
  \tikzstyle{transition}=[rectangle,thick,draw=black!75,fill=black!20,minimum size=4mm]

\usepackage{hyperref}

\begin{document}

\setcounter{page}{97}
\publyear{2021}
\papernumber{2083}
\volume{183}
\issue{1-2}

  \finalVersionForARXIV

\title{Cost Problems for Parametric Time Petri Nets\thanks{This work has been partially funded by ANR project ProMiS number ANR-19-CE25-0015}}

\author{Didier Lime\thanks{Address for correspondence: \'Ecole Centrale de Nantes, LS2N UMR CNRS 6004, Nantes, France},
~Olivier H. Roux
\\
\'Ecole Centrale de Nantes\\
LS2N UMR CNRS 6004, Nantes, France\\
\{Didier.Lime, Olivier-h.Roux\}@ec-nantes.fr
\and
Charlotte Seidner\\
Université de Nantes\\
LS2N UMR CNRS 6004, Nantes, France\\
Charlotte.Seidner@univ-nantes.fr
}

\maketitle

    \runninghead{D. Lime et al.}{Cost Problems for Parametric Time Petri Nets}

\vspace*{-6mm}
\begin{abstract}
We investigate the problem of parameter synthesis for time Petri nets with a cost variable that evolves both continuously with time, and discretely when firing transitions. More precisely, parameters are rational symbolic constants used for time constraints on the firing of transitions and we want to synthesise all their values such that some marking is reachable, with a cost that is either minimal or simply less than a given bound.

We first prove that the mere existence of values for the parameters such that the latter property holds is undecidable. We nonetheless provide symbolic semi-algorithms for the two synthesis problems and we prove them both sound and complete when they terminate. We also show how to modify them for the case when parameter values are integers. Finally, we prove that these modified versions terminate if parameters are bounded. While this is to be expected since there are now only a finite number of possible parameter values, our algorithms are symbolic and thus avoid an explicit enumeration of all those values. Furthermore, the results are symbolic constraints representing finite unions of convex polyhedra that are easily amenable to further analysis through linear programming.

We finally report on the implementation of the approach in Romeo, a software tool for the analysis of time Petri nets.
\end{abstract}


\section{Introduction}  \label{sec:intro}
So-called \emph{priced} or \emph{cost timed} models are suitable for representing real-time systems whose behaviour is constrained by some resource consuming (be it energy or CPU time, for instance) and for which we need to assess the total cost accumulated during their execution.
Such models can even describe whether the evolution of the cost during the run is caused by staying in a given state (continuous cost) or by performing a given action (discrete cost).
Thus, the task of finding if the model can reach some ``good'' states while keeping the overall cost under a given bound (or, further, finding the minimum cost) can prove of interest in many real-life applications, such as optimal scheduling or production line planning.

Timed models, however, require a thorough knowledge of the system for their analysis and are thus difficult to build in the early design stages, when the system is not fully identified.
Even when all timing constraints are known, the whole design process must often be carried out afresh, whenever the environment changes.
To obtain such valuable characteristics as flexibility and robustness, the designer may want to relax constraints on some specifications by allowing them a wider range of values.
To this end, parametric reasoning is particularly relevant for timed models, since it allows designers to use parameters instead of definite timing values.

We therefore propose to tackle the definition and analysis of models that support both (linear) cost functions and timed parameters.\medskip

\noindent  \textbf{Related work}~~~
Parametric timed automata (PTA) \cite{alur-acm-93}
extend timed automata \cite{alur-tcs-94} to overcome the limits of
checking the correctness of the systems with respect to definite timing
constraints.
The reachability-emptiness problem, which tests whether there exists a parameter valuation such that the automaton has an accepting run, is fundamental to any verification process but is undecidable~\cite{alur-acm-93}.
L/U automata \cite{hune-jlap-02} use each parameter either as a lower bound or as an upper bound on clocks. The reachability-emptiness problem is decidable for this model, but the state-space exploration, which would allow for explicit synthesis of all the suitable parameter valuations, still might not terminate~\cite{jovanovic-TSE-15}.
To obtain decidability results, the approach described in~\cite{jovanovic-TSE-15} does not rely on syntactical restrictions on guards and invariants, but rather on restricting the parameter values to bounded integers.
From a practical point of view, this subclass of PTA is not that restrictive, since the time constraints of timed automata are usually expressed as natural (or perhaps rational) numbers.

In \cite{ALUR2004297}, the authors have proved the decidability of the optimal-cost problem for Priced Timed Automata  with non-negative costs. In \cite{Behrmann2001,Behrmann2005,Larsen01ascheap}, the computation of the optimal-cost to reach a goal location is based on a forward exploration of zones extended with linear cost functions.
In \cite{bouyer-CAV-16}, the authors have improved this approach, so as to ensure termination of the forward exploration algorithm, even when clocks are not bounded and costs are negative, provided that the automaton has no negative cost cycles.
In \cite{DBLP:journals/corr/AbdullaM13}, the considered model is a timed arc Petri net, under weak firing semantics, extended with rate costs associated with places and firing costs associated with transitions. The computation of the optimal-cost for reaching a goal marking is based on similar techniques to~\cite{ALUR2004297}. In~\cite{boucheneb-FORMATS-17}, the authors have investigated the optimal-cost reachability problem for time Petri nets where each transition has a firing cost and each marking has a rate cost (represented as a linear rate cost function over markings). To compute the optimal-cost  to reach a goal marking, the authors have revisited the state class graph method to include costs.

\paragraph{Our contribution}
We propose in Section~\ref{sec:costtpn} an extension of time Petri nets with costs (both discrete and continuous with time) and timing parameters, i.e.,
rational symbolic constants used in the constraints on the firing times of transitions.

Within this formalism, we define three problems dealing with parametric reachability with cost constraints.
We prove  in Section~\ref{sec:npcosts} that the existence of a parameter valuation to reach a given marking under a given bounded cost is undecidable. This proof adapts a 2-counter machine encoding first proposed in~\cite{aleksandra-phd-13} for PTA. To our knowledge it is the first time a direct Petri net encoding is provided and the adaptation is not trivial.
We give in Section~\ref{sec:semialgo} a symbolic semi-algorithm that computes all such parameter valuations when it terminates, and we prove its correctness. We also provide another symbolic semi-algorithm that computes all parameter valuations such that a given marking is reachable with a minimal cost (over all runs and all parameter valuations), with corresponding completeness and soundness proofs.
We  propose in Section~\ref{sec:integer} variants of these semi-algorithms that compute integer parameter valuations and prove in Section~\ref{sec:termination} their termination provided parameter valuations are a priori bounded and the cost of each run is uniformly lower-bounded for integer parameter valuations.  This technique is symbolic and avoids the explicit enumeration of all possible parameter valuations. The basic underlying idea of using the integer hull operator was first investigated in~\cite{jovanovic-TSE-15} for PTA, but this is the first time that it is adapted and proved to work with state classes for time Petri nets, and the fact that it naturally also preserves costs for integer parameter valuations is new and very interesting.
We finally describe in Section~\ref{sec:casestudy} the implementation of the approach in the tool Romeo
by analysing a small scheduling case-study. This article is an extension of \cite{lime-ICATPN-19}, with mainly the addition of the optimal synthesis problem, algorithm, and implementation in Romeo.

\section{Parametric cost time Petri nets}
\label{sec:costtpn}
\subsection{Preliminaries}

    We denote the set of natural numbers (including $0$) by $\N$, the set of integers by $\Z$, the set of rational numbers by $\Q$ and the set of real numbers by $\R$. We note $\Qp$ (resp. $\Rp$) the set of non-negative rational (resp. real) numbers.
	For $n \in \N$, we let $\llbracket 0, n \rrbracket$ denote the set $\{ i \in \N \mid i \leq n \}$.
    For a finite set $X$, we denote its size by $|X|$.

\smallskip
    Given a set $X$, we denote by $\Intervals{X}$, the set of non empty real intervals that have their finite end-points in $X$.
    For $I \in \Intervals{X}$, $\eft{I}$ denotes its left end-point if $I$ is left-bounded and $-\infty$ otherwise. Similarly, $\lft{I}$ denotes the right end-point if $I$ is right-bounded and $\infty$ otherwise.
    We say that an interval $I$ is non-negative if $I\subseteq \Rp$.
	Moreover, for any non-empty non-negative interval $I$ and any $d \in \Rp$ such that $d\leq d'$ for some $d'\in I$, we let $I\minusI d$ be the interval defined by $\{\theta-d\mid\theta\in I\wedge \theta-d\geq0\}$. Note that this is again a non-empty non-negative interval.
	
\smallskip
    Given sets $V$ and $X$, a $V$-valuation (or simply valuation when $V$ is clear from the context) of $X$ is a mapping from $X$ to $V$. We denote by $V^X$ the set of $V$-valuations of $X$. When $X$ is finite, given an arbitrary fixed order on $X$, we often equivalently consider $V$-valuations as vectors of $V^{|X|}$.
    Given a $V$-valuation $v$ of $X$ and $Y\subseteq X$, we denote by $\proj{v}{Y}$ the projection of $v$ on $Y$, i.e., the valuation on $Y$ such that $\forall x\in Y, \proj{v}{Y}(x)=v(x)$.

\subsection{Time Petri nets with costs and parameters}\smallskip

\begin{definition}[Parametric Cost Time Petri Net (pcTPN)]
    A \emph{Parametric Cost Time Petri Net} (pcTPN) is a tuple $\Net = (P,T,\Params,\pre{.},\post{.},m_0,I_s,\discreteCost,\costrate)$ where:
	\begin{itemize}
\itemsep=0.9pt
		\item $P$ is a finite non-empty set of \emph{places},
		\item $T$ is a finite set of \emph{transitions} such that $T \cap P = \emptyset$,
        \item $\Params$ is a finite set of \emph{parameters},
		\item $\pre{.}: T \rightarrow \Markings$ is the backward incidence mapping,
		\item $\post{.}: T \rightarrow \Markings$ is the forward incidence mapping,
		\item $m_0\in \Markings$ is the initial marking,
        \item $I_s: T\rightarrow \Intervals{\N\cup\Params}$ is the (parametric) \emph{static firing interval} function,
        \item $\discreteCost: T \rightarrow \Z$ is the \emph{discrete cost} function, and
        \item $\costrate: \Markings \rightarrow \Z$ is the \emph{cost rate} function.
	\end{itemize}
\end{definition}

    Given a parameterized object $x$ (be it a pcTPN, a function, an expression, etc.), and a $\Q$-valuation $v$ of parameters, we denote by $v(x)$ the corresponding non-parameterized object, in which each parameter $a$ has been replaced by the value $v(a)$.

    A \emph{marking} is an $\N$-valuation of $P$.
    For a marking $m\in \Markings$, $m(p)$ represents a number of \emph{tokens} in place $p$.
    A transition $t \in T$ is said to be \emph{enabled} by a given marking $m \in \Markings$ if for all places $p$, $m(p)\geq \pre{t}(p)$. We also write $m\geq \pre{t}$.
	We denote by $\en{m}$ the set of transitions that are enabled by the marking $m$: $\en{m} = \{t \in T \mid m \geq \pre{t}\}$.

    Firing an enabled transition $t$ from marking $m$ leads to a new marking $m' = m-\pre{t}+\post{t}$.
    A transition $t' \in T$ is said to be \emph{newly enabled} by the firing of a transition $t$ from a given marking $m \in \Markings$ if it is enabled by the new marking but not by $m-\pre{t}$ (or it is itself fired)\footnote{With that definition we use the classical ``intermediate semantics'' of Berthomieu~\cite{berthomieu-TSE-91}. A study of alternative semantics can be found in~\cite{berard-ATVA-05}.}.

    We denote by $\newen{m,t}$ the set of transitions that are newly enabled by the firing of $t$ from the marking $m$:
    $\newen{m,t} =\big \{t' \in \en{m-\pre{t}+\post{t}} \; | \; t'\not\in \en{m-\pre{t}} \text{ or } t=t'\big\}$

    A \emph{state} of the net $\Net$ is a tuple $(m,I,c,v)$ in $\states\times\R\times\Qp^\Params$, where: $m$ is a marking of $\Net$, $I$ is called the interval function and associates a \emph{temporal interval} to each transition enabled by $m$.
    Value $c$ is the cost associated with that state and valuation $v$ assigns a rational value to each parameter for the state.

\begin{definition}[Semantics of a pcTPN]
The semantics of a pcTPN is a timed transition system $(Q,Q_0,\rightarrow)$ where:
\begin{itemize}
    \item $Q\subseteq \states\times \R \times \Qp^\Params$
    \item $Q_0=\{(m_0,I_0, 0, v) | v\in\Qp^\Params, \forall t\in T, v(I_s(t))\neq\emptyset\}$ where $\forall t\in \en{m_0}, I_0(t)=I_s(t)$
    \item $\rightarrow$ consists of two types of transitions:
    \begin{itemize}
        \item discrete transitions: $(m,I,c,v)\xrightarrow{t\in T} (m',I',c',v)$ iff
            \begin{itemize}
                \item $m\geq \pre{t}$, $m'=m - \pre{t} +\post{t}$  and $\eft{I(t)} = 0$,
                \item $\forall t'\in \en{m'}$
	            \begin{itemize}
                    \item   $I'(t') = v(I_s(t'))$ if $t'\in\newen{m,t}$,
			        \item  $I'(t')=I(t')$ otherwise
	    	    \end{itemize}
                \item $c' = c + \discreteCost(t)$
    	  \end{itemize}
        \item time transitions: $(m,I,c,v)\xrightarrow{d\in \Rp} (m,I \minusI d,c',v)$,
            iff $\forall t\in \en{m}$, $\exists d'\in I(t) \text{ s.t. } d\leq d'$ and $c' = c + \costrate(m)*d$.
    \end{itemize}
\end{itemize}
\label{def:sem-pctpn}
\end{definition}

A run of a pcTPN $\Net$ is a possibly infinite sequence $q_0a_0q_1a_1q_2a_2\cdots$ such that $q_0\in Q_0$, for all $i>0$, $q_i\in Q$, $a_i\in T\cup\Rp$ and $q_i\xrightarrow{a_i} q_{i+1}$.
The set of runs of $\Net$ is denoted by $\Runs(\Net)$.
We note $(m,I,c,v)\xrightarrow{t@d}(m',I',c',v)$ for the sequence of elapsing $d\geq 0$ followed by the firing of the transition $t$.
We denote by $\sequence(\rho)$ the projection of the run $\rho$ over $T$: for a run $\rho=q_0\xrightarrow{t_0@d_0}q_1\xrightarrow{t_1@d_1}q_2\xrightarrow{t_2@d_2}q_3\xrightarrow{t_3@d_3}\cdots$, we have $\sequence(\rho) =t_0t_1t_2t_3\cdots$.
We write $q\xhookrightarrow{t}q'$ if there exists $d\geq 0$ such that $q\xrightarrow{t@d} q'$.

For a finite run $\rho$ we denote by $\last(\rho)$ the last state of $\rho$ and by $\lastm(\rho)$ its marking.
A state $(m,I,c,v)$ is said to be \emph{reachable} if there exists a finite run $\rho$ of the net, with $\last(\rho)=(m,I,c,v)$. A marking $m$ is reachable for parameter valuation $v$, if there exists some $I$ and $c$ such that $(m, I, c, v)$ is reachable.

For $k\in\N$ and parameter valuation $v$, the (Cost) Time Petri net $v(\Net)$ is said to be $k$-bounded if for all reachable markings $m$, and all places $p$, $m(p)\leq k$. We say that $v(\Net)$ is bounded if there exists $k$ such that it is $k$-bounded.

The \emph{cost} $\costrun(\rho)$ of a finite run $\rho$, with last state $(m,I,c,v)$ is $c$.

Since we are interested in minimising the cost, the \emph{cost} of a sequence of transitions $\sigma$ is defined as $\costclass(\sigma)= \inf_{\rho\in\Runs(\Net), \sequence(\rho)=\sigma} \costrun(\rho)$.

For the sake of the clarity of the presentation, we consider only closed intervals (or right-open to infinity) so this infimum is actually a minimum.

        \begin{figure}[h]
            \centering
            \scalebox{0.98}{
            \begin{tikzpicture}[auto]
                \node[normal place,tokens=1,label={above:$p_0$}] (p0) {};
                \node[normal place,below =of p0,tokens=1,label={below:$p_1$}] (p1) {};

                \node[transition, right=of p0,label={above:$\begin{array}{c}\cost(t_0)=2\\{}[a,a]\end{array}$}] (t0) {$t_0$};
                \draw (t0) edge[pre,bend right=20] (p0)
                           edge[post,bend left=20] (p0);

                \node[transition, right=of p1,label={below:$[2,5]$}] (t1) {$t_1$};

                \node[normal place,right=of t1,tokens=0,label={below:$p_2$}] (p2) {};
                \draw (t1) edge[pre] (p1) edge[post] (p2);

                \node [right=1cm of t0] {$\cost(\vec{m})=2m(p_0) + m(p_1)$};
            \end{tikzpicture} }\vspace*{-2mm}
            \caption{A parametric cost time Petri net.}
            \label{fig:pctpn}
        \end{figure}
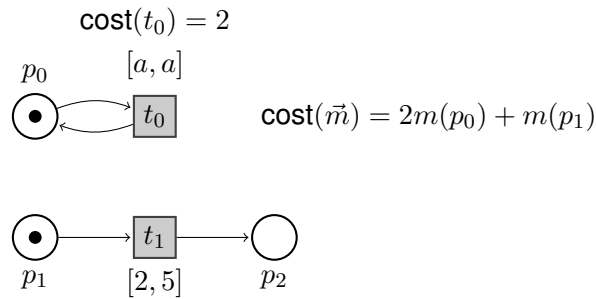

\begin{example}
    Consider the net in Figure~\ref{fig:pctpn}. It has one parameter $a$, and the cost increases with time with a rate equal to twice the number of tokens in $p_0$ + once the number of tokens in $p_1$.  It also increases by $2$ when firing $t_0$ and does not change when firing $t_1$ (so we omit it in the figure).
    For parameter value $v(a)=2$, the initial state is $(\{p_0,p_1\}, I(t_0)=[2,2], I(t_1)=[2,5], 0, 2)$. We write the marking as the set of marked places because the net is safe to make it simpler. Cost is initially $0$.

\smallskip
    Then a possible run is:\vspace*{-2mm}
        $$
        \hspace*{-2mm}\begin{array}{ll}(\{p_0,p_1\}, I(t_0)\!=[2,2], I(t_1)=[2,5], 0, 2)
            &\hspace*{-2mm}\xrightarrow{2} (\{p_0,p_1\}, I(t_0)\!=[0,0], I(t_1)=[0,3], 6, 2)\\
            &\hspace*{-2mm}\xrightarrow{t_0} (\{p_0,p_1\}, I(t_0)\!=[2,2], I(t_1)=[0,3], 8, 2)\\
            & \hspace*{-2.4mm}\xrightarrow{0.2} (\{p_0,p_1\}, I(t_0)\!=[1.8,1.8], I(t_1)=[0,2.8], 8.6, 2)\\
            & \hspace*{-2mm}\xrightarrow{t_1} (\{p_0,p_2\}, I(t_0)\!=[1.8,1.8], 8.6, 2)
        \end{array}\hspace*{-2mm}
        $$
\end{example}

\subsection{Parametric cost problems}

Given a set of target markings $\aGoal$, the problems we are interested in are:
\begin{enumerate}
\itemsep=0.9pt
    \item the existential problem: Given a finite maximum cost value $\cmax$, is there a parameter valuation $v$ such that some marking in $\aGoal$ is reachable with a cost less than $\cmax$ in $v(\Net)$?
    \item the bounded synthesis problem:$\,$Given a finite maximum cost value$\,\cmax\,$,compute all the parame\-ter valuations$\,v\,$such that some
          marking in$\,\aGoal\,$is reachable with a cost less than$\,\cmax\,$in$\,v(\Net)$.
    \item the optimal synthesis problem: Compute the infimum of the cost of all runs that reach a marking in $\aGoal$, $\inf_{m\in \aGoal, v\in\Params^\Q,\rho\in \Runs(\Net), \lastm(\rho)\in \aGoal} \costrun(\rho)$, and all the parameter valuations $v$ such that this infimum cost can be achieved in $v(\Net)$,
\end{enumerate}

\begin{example}
    For the net in Figure~\ref{fig:pctpn}, there is no parameter valuations allowing to reach $p_2=1$ with a cost less or equal to $5$, because we need to wait at least $2$ time units (t.u.) before firing $t_1$ with a cost rate of $3$. For a bound equal to $8$, it will be $1\leq a$ so that $t_0$ fires at most once before $t_1$.

    The minimum cost to reach $p_2=1$ is $6$ and this can be achieved iff $2\leq a$, so that $t_0$ is not forced to fire before $t_1$.
\end{example}

We prove in Section~\ref{sec:npcosts} that the existential problem is undecidable.

\section{Undecidability results}
\label{sec:npcosts}

 The existential parametric time bounded reachability problem for bounded parametric time Petri nets asks whether a given target marking is reachable for some valuation of the parameter(s) within $\cmax$ time units. This is a special case of the existential cost bounded reachability problem defined in Section~\ref{sec:costtpn}, with no discrete cost and a uniform cost rate of $1$. Proposition~\ref{TBR-undecidable} therefore implies the undecidability of that more general problem.

 As in~\cite{alur-acm-93}, our proof is based on a blocking variant of classical two counter machines. We first prove the undecidability of the halting problem for this variant in Lemma~\ref{blockingM2C}, since this has to the best of our knowledge never been done formally.

\begin{lemma}
    Consider a non-deterministic counter machine with two non-negative counters $C_1$ and $C_2$, control states $\{s_1,\ldots,s_n\}$, and, for $i,j\in\{1,\ldots,n\}$, and $x\in\{1,2\}$, instructions of the form:
    \begin{enumerate}
    \itemsep=0.9pt
        \item in state $s_i$, increment counter $C_x$, and go to state $s_j$;
        \item in state $s_i$, if $C_x>0$ decrement $C_x$ and go to state $s_j$ else block;
        \item in state $s_i$, if $C_x=0$ then go to state $s_j$ else block.
    \end{enumerate}

    Assume the machine halts when reaching $s_n$. Then the halting problem for such a machine, i.e., knowing whether the machine will eventually halt or not, is undecidable.
    \label{blockingM2C}
\end{lemma}

\begin{figure}[h]
\vspace*{-2mm}
\centering
\resizebox{1\linewidth}{!}{
\begin{tikzpicture}[node distance=1cm,>=stealth',bend angle=45,auto]
    \node [grayclair place, tokens=0] (px1) [label={[black]above:$P_{x\leq b}$}] {};
    \node [brown place,tokens=0] (fx1) [yshift=-11mm,label={[black]left:$P_{x=b}$},below of=px1] {};
    \node [grayclair place,tokens=0] (pxa) [yshift=-33mm,label={[black]left:$P_{x\geq b}$},below of=px1] {};
    \node [brown place, tokens=0] (fxa) [yshift=-13mm,label={[black]left:$P_{x=a+b}$},below of=pxa] {};

    \node [grayclair place,tokens=0] (py1) [xshift=+55mm,label={[black]above:$P_{y\leq b}$},right of=px1] {};
    \node [brown place,tokens=0] (fy1) [yshift=-11mm,label={[black]right:$P_{y=b}$},below of=py1] {};
    \node [grayclair place,tokens=0] (pya) [yshift=-33mm,label={[black]right:$P_{y\geq b}$},below of=py1] {};
    \node [brown place, tokens=0] (fya) [yshift=-13mm,label={[black]right:$P_{y=a+b}$},below of=pya] {};

   \node [brown place,tokens=0] (pz0) [xshift=-15mm,yshift=-12mm,label={[black]below:$P_{z=0}$},below of=fxa] {};
    \node [grayclair place,tokens=0] (pz1) [xshift=11mm,label={[black]below:$P_{z\leq b}$},right of=pz0] {};
    \node [brown place,tokens=0] (fz1) [xshift=13mm,label={[black]below:$P_{z=b}$},right of=pz1] {};
    \node [grayclair place,tokens=0] (pza) [xshift=33mm,label={[black]below:$P_{z\geq b}$},right of=pz1] {};
    \node [brown place, tokens=0] (fza) [xshift=16mm,label={[black]below:$P_{z=a+b}$},right of=pza] {};


  \node [normal place,tokens=0] (incrCy) [xshift=18mm,yshift=-0mm,label={[black]above:$s_i$},right of=px1] {};
  \node [blue place,tokens=0] (checkx1) [xshift=-9mm,yshift=-13mm,label={[black]right:$P_x$},below of=incrCy] {};
  \node [purple place,tokens=0] (checkya) [xshift=9mm,yshift=-13mm,label={[black]left:$P_y$},below of=incrCy] {};
  \node [blue place,tokens=0] (okx) [xshift=0mm,yshift=-13mm,label={[black]left:$ok_x$},below of=checkx1] {};
  \node [purple place,tokens=0] (oky) [xshift=0mm,yshift=-13mm,label={[black]right:$ok_y$},below of=checkya] {};
  \node [normal place,tokens=0] (endCy) [xshift=0mm,yshift=-65mm,label={[black]below:$s_j$},below of=incrCy] {};

   \node [yellow transition] (emptyz0) [right of=pz0, label={[black]below:$[0,0]$}] {$\varepsilon_{z_{0}}$}
      edge [pre] (pz0)
      edge [post] (pz1);
   \node [normal transition] (tz1) [right of=pz1, label={[black]below:$[b,b]$}] {$t_{z=b}$}
      edge [pre] (pz1)
      edge [post] (fz1);
   \node [yellow transition] (emptyz1) [right of=fz1, label={[black]below:$[0,0]$}] {$\varepsilon_{z_{b}}$}
      edge [pre] (fz1)
      edge [post] (pza);
   \node [normal transition] (tza) [xshift=2mm, right of=pza, label={[black]below:$[a,a]$}] {$t_{z=a+b}$}
      edge [pre] (pza)
      edge [post] (fza);
   \node [yellow transition] (emptyza) [right of=fza, label={[black]below:$[0,0]$}] {$\varepsilon_{z_{a+b}}$}
      edge [pre] (fza);

    \node [normal transition] (txa) [below of=pxa, label={[black]left:$[a,a]$}] {$t_{x={a+b}}$}
      edge [pre] (pxa)
      edge [post] (fxa);
    \node [normal transition] (tx1) [below of=px1, label={[black]left:$[b,b]$}] {$t_{x=b}$}
      edge [pre] (px1)
      edge [post] (fx1);

    \node [normal transition] (tya) [below of=pya, label={[black]right:$[a,a]$}] {$t_{y=a+b}$}
      edge [pre] (pya)
      edge [post] (fya);
    \node [normal transition] (ty1) [below of=py1, label={[black]right:$[b,b]$}] {$t_{y=b}$}
      edge [pre] (py1)
      edge [post] (fy1);

          \node [yellow transition] (emptyxa) [below of=fxa, label={[black]left:$[0,0]$}] {$\varepsilon_{x_{a+b}}$}
      edge [pre] (fxa);
    \node [yellow transition] (emptyx1) [below of=fx1, label={[black]left:$[0,0]$}] {$\varepsilon_{x_b}$}
      edge [pre] (fx1)
      edge [blue,pre] (okx)
      edge [post] (pxa);
    \node [yellow transition] (emptyya) [below of=fya, label={[black]right:$[0,0]$}] {$\varepsilon_{y_{a+b}}$}
      edge [pre] (fya);
    \node [yellow transition] (emptyy1) [below of=fy1, label={[black]right:$[0,0]$}] {$\varepsilon_{y_b}$}
      edge [pre] (fy1)
      edge [post] (pya);

    \node [green transition] (startIncrCy) [below of=incrCy,xshift=0mm,yshift=0mm, label={[blue]right:$[0,0]$}] {$start$}
      edge [pre] (incrCy)
      edge [purple,post] (checkya)
      edge [blue,post] (checkx1);

   \node [purple transition] (ya) [below of=checkya, label={[purple]right:$[0,0]$}] {$R(y)$}
      edge [purple,pre] (fya)
      edge [purple,pre] (checkya)
      edge [purple,post] (oky)
      edge [purple,post] (py1);

  \node [blue transition] (x1) [below of=checkx1, label={[black]left:$[0,0]$}] {$R(x)$}
      edge [blue,pre] (fx1)
      edge [blue,pre] (checkx1)
      edge [blue,post] (okx)
      edge [blue,post] (px1);

    \node [green transition] (doneCy) [below of=startIncrCy,xshift=0mm,yshift=-39mm, label={[blue]right:$[0,0]$}] {$done$}
      edge [blue,pre] (okx)
      edge [purple,pre] (oky)
      edge [pre,bend left=20] (fz1)
      edge [post] (endCy)
      edge [post,bend left=20] (pz0);

    \node [grayclair place,tokens=1] (pp0) [xshift=110mm,label={[black]above:$P_0$},right of=px1] {};
    \node [grayclair place,tokens=1] (pp0a) [label={[black]above:$P_{loop_a}$},right of=pp0] {};
    \node [grayclair place,tokens=0] (pp1) [yshift=-13mm,label={[black]above left:$P_{1}$},below of=pp0] {};
    \node [grayclair place,tokens=0] (ppa) [yshift=-13mm,label={[black]above left:$P_{a}$},below of=pp0a] {};

    \node [normal transition] (loop) [right of=pp0a, label={[black]above:$[0,a]$}] {$loop_a$}
      edge [pre,bend left] (pp0a)
      edge [post,bend right] (pp0a);

    \node [grayclair place,tokens=1] (pp0b) [label={[black]above:$P_{loop_b}$},right of=loop,xshift=5mm] {};
    \node [grayclair place,tokens=0] (ppb) [yshift=-13mm,label={[black]above left:$P_{b}$},below of=pp0b] {};

   \node [normal transition] (loopb) [right of=pp0b, label={[black]above:$[0,b]$}] {$loop_b$}
      edge [pre,bend left] (pp0b)
      edge [post,bend right] (pp0b);

   \node [normal transition] (ta) [below of=pp0a, label={[black]right:$[a,a]$}] {$t_a$}
      edge [pre] (pp0a)
      edge [post] (ppa);

   \node [normal transition] (tb) [below of=pp0b, label={[black]right:$[b,b]$}] {$t_b$}
      edge [pre] (pp0b)
      edge [post] (ppb);

   \node [normal transition] (t1) [below of=pp0, label={[black]left:$[1,1]$}] {$t_1$}
      edge [pre] (pp0)
      edge [post] (pp1);

   \node [normal transition] (emptya) [right of=ppa, label={[black]above:$[0,0]$}] {$\varepsilon_a$}
      edge [pre] (ppa) ;
   \node [normal transition] (emptya) [right of=ppb, label={[black]above:$[0,0]$}] {$\varepsilon_b$}
      edge [pre] (ppb) ;

    \node [normal transition] (empty1) [left of=pp1, label={[black]below:$[0,0]$}] {$\varepsilon_1$}
      edge [pre] (pp1) ;

    \node [brown place,tokens=0] (ppx1) [xshift=-5mm,yshift=-11mm,label={[black]below:$P_{x=b}$},below left of=ppa] {};
    \node [brown place,tokens=0] (ppy1) [yshift=-13mm,xshift=-5mm,label={[black]below:$P_{y=b}$},below  of=ppa] {};
    \node [brown place,tokens=0] (ppz0) [yshift=-13mm,xshift=5mm,label={[black]below:$P_{z=0}$},below of=ppa] {};
    \node [brown place,tokens=0] (ppstart) [yshift=-11mm,xshift=5mm,label={[black]right:$s_0$},below right of=ppa] {};

    \node [normal transition] (tstart) [below of=ppa, label={[black]right:$[0,0]$}] {$start$}
      edge [pre] (ppa)
      edge [pre] (ppb)
      edge [pre] (pp1)
      edge [post] (ppx1)
      edge [post] (ppy1)
      edge [post] (ppz0)
      edge [post] (ppstart);

          \node  (figb) at (14,-6) {\ref{incrCywithb}.b) Initialise the parameters $a$ and $b$};

          \node  (figb2) at (14.6,-6.4) {such that: $0<a\leq1$ and $0<b\leq1$};

        \node  (figa) at (12,-8.6) {\ref{incrCywithb}.a) Encoding of: when in state $s_i$, };

         \node  (figa) at (12.6,-9) {increment  $C_y$ and go to $s_j$};
\end{tikzpicture}\vspace*{-2mm}
}
\caption{ Increment gadget (left) and initial gadget (right) }
\label{incrCywithb}\vspace*{-3mm}
\end{figure}
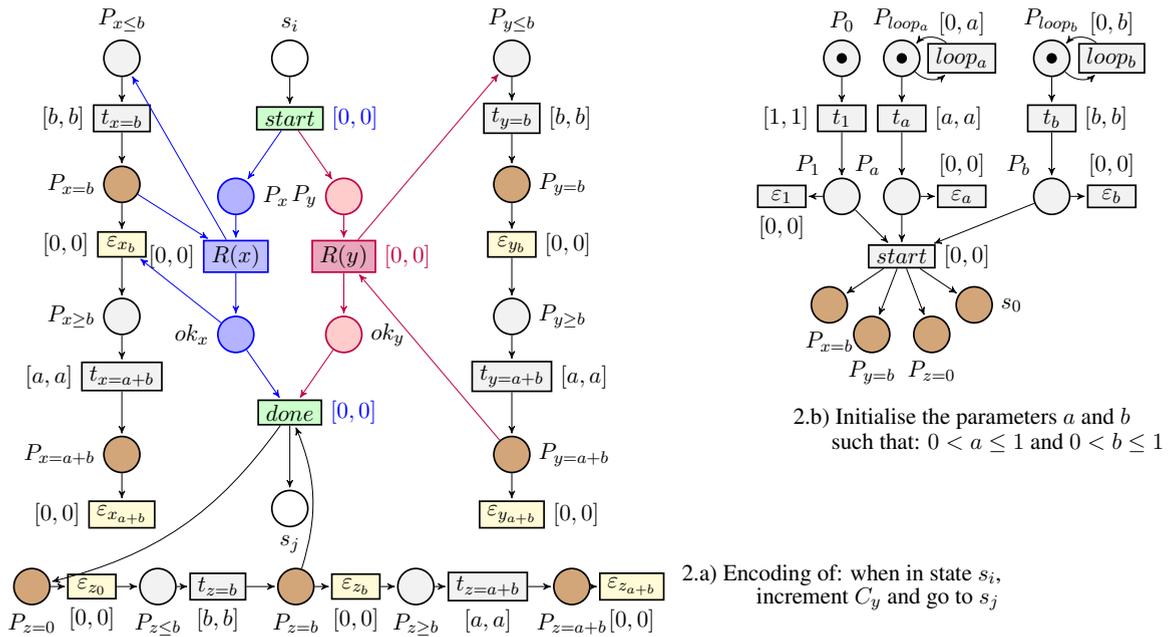

\begin{proof}
    Consider a classical deterministic Minsky machine~\cite{Minsky67}, instead of our decrement and zero test it has instructions of the form: ``in state $s_i$, if $C_x>0$ decrement $C_x$ and go to state $s_j$ else go to $s_k$''.

    We can simulate such an instruction by a non-deterministic choice between instructions:
        ``in state $s_i$, if $C_x>0$ decrement $C_x$ and go to state $s_j$ else block''
        and
        ``in state $s_i$, if $C_x=0$ then go to state $s_k$ else block''.

    It is clear that the resulting machine has exactly one non-blocking path, corresponding to the path of the classical deterministic Minsky machine. Since the halting problem is undecidable for the latter, it is therefore undecidable for our blocking machines.
\end{proof}

\begin{proposition}
\label{TBR-undecidable}
    Existential parametric time bounded reachability is undecidable for bounded parametric time Petri nets.
\end{proposition}

\begin{proof}
  Given a bounded parametric time Petri net $\Net$, we want to decide whether there exists some parameter valuation $v$ such that some given marking can be reached within $\cmax$ time units in $v(\Net)$. The idea of this proof was first sketched in~\cite{aleksandra-phd-13} for parametric timed automata.
  We encode the halting problem for two-counter machines, which is undecidable~\cite{Minsky67}, into the existential problem for parametric time Petri nets.

    We consider the 2-counter machine variant described in Lemma~\ref{blockingM2C}.
The machine starts in state $s_0$ and halts when it reaches a particular state $s_{halt}$.

Given such a machine $\cal{M}$, we now provide an encoding as a parametric time Petri net $\cal{N_M}$: each state $s_i$ of the machine is encoded as place, which we also call $s_i$.
  The encoding of the 2-counter machine $\cal{M}$ is as follows: it uses two rational-valued parameter $a$ and $b$,
    and three gadgets  shown in Figure~\ref{incrCywithb}.a modelling three clocks $x$, $y$, $z$.
    Provided all intervals are closed, which is the case in this proof, for a state $(m,I,c,v)$, the duration for which an enabled transition $t$ has been continuously enabled, which we call \emph{enabling time} of $t$, is $v(\lft{I_s(t)}-\lft{I(t)})$.
    For the gadget modelling clock $x$, the value of clock $x$ is equal to:
  i) the enabling time of transition $t_{x=b}$ when $P_{x\leq b}$ is marked;
    ii)  $b$ when $P_{x= b}$ is marked;
    iii) the sum of $b$ and the enabling time of transition $t_{x=a+b}$ when $P_{x\geq b}$ is marked (note that this value is lower than $a+b$);
   iv) $a+b$    when $P_{x=a+b}$ is marked;
   v) an unknown (and irrelevant) value in all other cases.
The gadget encoding the increment instruction of $C_y$ is given in Figure~\ref{incrCywithb}.a.
Clocks $x$ and $y$ store the value of each counter $C_x$ and $C_y$ as follows $x =b-aC_x$ and $y =b-aC_y$ when $z=0$.
The  zero-test gadget is given in Figure~\ref{BD:zero-test}.
The system is studied over 2 time units, one of which is spent in the initialisation gadget.

    \paragraph{Initialisation} We use the gadget in~Figure \ref{incrCywithb}.b to initialise $a$ and $b$ such that $0<a \leq 1$ and $0<b \leq 1$.
We explain how this works for $a$, the case of $b$ is similar. First, if $a=0$ then $t_a$ must be fired at date $0$ and the token produced in $P_a$ will be consumed by transition $\varepsilon_a$ before transition $start$ has a chance to fire. Now, if we do not want the token produced in $P_1$ at date $1$ to be consumed for sure by transition $\varepsilon_1$, transition $t_a$ must fire at date $1$ at the latest, and therefore $a$ must be less than or equal to $1$.
   Finally, when $0<a\leq 1$, a sufficient number of iterations of transition $loop_a$ allows a token to be produced in $P_a$ exactly at date $1$: it is for instance always possible to choose some duration $d\leq a$, spent before each iteration of $loop_a$ and a number of iterations $n$ such that $nd = 1-a$.
    All in all, transition $start$ can thus be fired if and only if $0<a\leq 1$ and $0<b\leq 1$.

\paragraph{Increment}
We start from some encoding configuration: $x=b-aC_x$, $y=b-aC_y$ and $z=0$ in a marking such that places $P_{z=0}$ and $s_i$ are marked. %
After the firing of transition $start$, there is an interleaving of transitions $R(x)$ and $R(y)$ that go through the gadget.
    Suppose that $1 < C_x < C_y$ (the other cases are similar). It follows that for $x$ and $y$, places $P_{x\leq b}$ and $P_{y\leq b}$ are marked.
    Then the next transition to fire is $t_{x=b}$ and when it fires we have $x=b$, $y=b-aC_y + aC_x$ and $z=a.C_x$.
    Then $R(x)$ fires, which corresponds to resetting $x$ to $0$: the token is back in $P_{x\leq b}$ and transition $t_{x=b}$ has been enabled for $0$ t.u.
    Then $t_{y=b}$ fires  (because $aC_y \leq b$), and then we have $x=aC_y - aC_x$, $y=b$, and $z=aC_y$.
    Then $t_{y=a+b}$ fires, and we have $x=aC_y - aC_x + a$, $y=a+b$ and $z=aC_y+a$.
    Then $R(y)$ fires, which corresponds to resetting $y$ to $0$.
    Then $t_{z=b}$ fires and we have $x=b - aC_x$, $y=b-aC_y-a$ and $z=b$.
    Finally, $done$ fires which resets $z$ and we have $x=b - aC_x$, $y=b-a(C_y + 1)$ and $z=0$ as expected.
    Note that if $a(C_y +1)> b$, then $t_{z=b}$ must fire before $t_{y=a+b}$ and then $\varepsilon_{z_b}$ fires immediately and $done$ can never be fired.
    Therefore, $v(\cal{N_M})$ will block for all the parameter valuations $v$ that do not correctly encode the machine.\vspace*{-1mm}

\paragraph{Decrement} By replacing the arc from $P_{z=b}$ to $done$ by an arc from $P_{z=a+b}$ to $done$, the only difference with the previous reasoning is that the time to fire $done$ is increased by $a$. Then, when $C_x>0$, we obtain  $z=0$,  $x=b+a-aC_x=b-a(C_x-1)$ and $y=b-aC_y$ corresponding to the decrement of $C_x$.

When $C_x=0$, we must begin with $R(x)$, then the sequence is $t_{y=b},t_{y=a+b},R(y),t_{z=b}$ as before and we obtain $x=b$, $y=b-aCy-a$, $z=b$. Now, since $x=b$ but $P_{x\leq b}$ is still marked, and since $a>0$, $t_{x=b}$  must fire next, i.e., before $t_{z=a+b}$. And then for the same reasons, $\varepsilon_{x_b}$ must fire and consume the token in $ok_x$ which means that $done$ can never fire and the machine is blocked.

We can obtain symmetrically (by swapping $x$ and $y$) the increment of $C_x$ and the decrement of$\,C_y$.\vspace*{-1mm}

\paragraph{Zero-test of $x$}
We start again from some encoding configuration: $x=b-aC_x$, $y=b-aC_y$ and $z=0$ in a marking such that places $P_{z=0}$ and $s_i$ are marked. %
After the firing of transition $start$, there is an interleaving of transitions $R(x)$ and $R(y)$ that go through the gadget.
    Suppose that $0 < C_x < C_y$ (the cases where $C_x>0$ and $C_x\geq C_y$ are similar). As before $P_{x\leq b}$ and $P_{y\leq b}$ are marked.
    Then the next transition to fire is $t_{x=b}$ and when it fires we have $x=b$, $y=b-aC_y + aC_x$ and $z=a.C_x$.
    Then $R(x)$ fires and $x$ is reset to $0$.
    Then $t_{y=b}$ fires, and we have $x=aC_y - aC_x$, $y=b$, and $z=aC_y$.
    Then $R(y)$ fires and $y$ is reset to $0$.
    Then $t_{z=b}$ fires and we have $x=b - aC_x$, $y=b-aC_y$ and $z=b$.
    But then place $P_{x\leq b}$ is marked and not place $P_{x=b}$ so transition $zero$ cannot fire. Also $x<b$ because $C_x>0$ so transition $t_{x=b}$ cannot fire immediately. So the only firable transition is $notzero$, and firing it blocks the machine.

    Suppose now that $C_x = 0$ and $C_y>0$ (the case where both are $0$ is similar). Then $P_{x = b}$ and $P_{y\leq b}$ are marked.
    We then fire $R(x)$ so $x$ becomes $0$.
    Then we fire $t_{y=b}$, $R(y)$, and $t_{z=b}$ as before, giving $x=b$, $y=b-aC_y$ and $z=b$.
    But now $t_{x=b}$ is firable, so we fire it and get $P_{x=b}$ marked so $zero$ can now fire which resets $z$ and we finally have $x=b$, $y=b-aC_y$ and $z=0$ as expected.\vspace*{-1mm}

 \paragraph{Proving the equivalence}
Both the increment gadget and the zero-test gadget require $b$ time units, and the decrement gadget  requires $(a+b)$ time units.	
Since the system executes over $1$ time unit, for any value of $a > 0$ and $b>0$, the number of operations that the machine can perform is finite.
We  consider two cases:
\begin{enumerate}
\itemsep=0.9pt
\item Either the machine halts. Then both counters $C_x$ and $C_y$ are bounded. Let $c$ be their maximum value over the whole execution and let $m$ be the number of steps of the finite halting execution of the machine.
If $c=0$ then the machine is a sequence of $m$ zero-tests taking $m.b$ time units and the parametric Petri net $\cal{N_M}$ can go within $1$ time unit to a marking $m_{halt}$ if $0<a \leq 1$ and $0<b \leq \frac{1}{m}$.
If $c>0$, since an instruction requires at most $a+b$ time units, if $a+b \leq \frac{1}{m}$ and if $0<a\leq \frac{b}{c}$
then there exists a run that correctly simulates the machine, and eventually reaches $m_{halt}$ within $1$ time unit.

This set of valuations is non-empty: for example if $c = 0$, then we can choose $a=b=\frac{1}{m}$ and if $c > 0$, then, since $m\geq c$, we can choose $a=\frac{b}{m}$  and $b=\frac{1}{1+m}$ hence $a=\frac{1}{m(1+m)}$.
\eject
\item Or the machine does not halt.
A step requires at least $b$ time units then for any value $v$ of the parameters, after a maximum number of steps (at most $\frac{1}{b}$), one whole time unit will elapse
without $v(\cal{N_M})$ reaching $m_{halt}$.
\end{enumerate}

\vspace*{-7mm}
\end{proof}

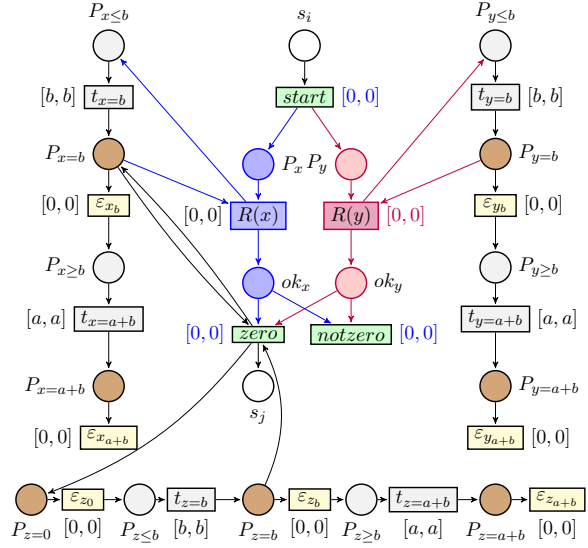
\begin{figure}[h]
\centering
\resizebox{.5\linewidth}{!}{
\begin{tikzpicture}[node distance=1.0cm,>=stealth',bend angle=45,auto]
     \node [grayclair place, tokens=0] (px1) [label={[black]above:$P_{x\leq b}$}] {};
    \node [brown place,tokens=0] (fx1) [yshift=-11mm,label={[black]left:$P_{x=b}$},below of=px1] {};
    \node [grayclair place,tokens=0] (pxa) [yshift=-33mm,label={[black]left:$P_{x\geq b}$},below of=px1] {};
    \node [brown place, tokens=0] (fxa) [yshift=-13mm,label={[black]left:$P_{x=a+b}$},below of=pxa] {};

    \node [grayclair place,tokens=0] (py1) [xshift=+65mm,label={[black]above:$P_{y\leq b}$},right of=px1] {};
    \node [brown place,tokens=0] (fy1) [yshift=-11mm,label={[black]right:$P_{y=b}$},below of=py1] {};
    \node [grayclair place,tokens=0] (pya) [yshift=-33mm,label={[black]right:$P_{y\geq b}$},below of=py1] {};
    \node [brown place, tokens=0] (fya) [yshift=-13mm,label={[black]right:$P_{y=a+b}$},below of=pya] {};

   \node [brown place,tokens=0] (pz0) [xshift=-15mm,yshift=-12mm,label={[black]below:$P_{z=0}$},below of=fxa] {};
    \node [grayclair place,tokens=0] (pz1) [xshift=11mm,label={[black]below:$P_{z\leq b}$},right of=pz0] {};
    \node [brown place,tokens=0] (fz1) [xshift=13mm,label={[black]below:$P_{z=b}$},right of=pz1] {};
    \node [grayclair place,tokens=0] (pza) [xshift=33mm,label={[black]below:$P_{z\geq b}$},right of=pz1] {};
    \node [brown place, tokens=0] (fza) [xshift=16mm,label={[black]below:$P_{z=a+b}$},right of=pza] {};


  \node [normal place,tokens=0] (incrCy) [xshift=28mm,yshift=-0mm,label={[black]above:$s_i$},right of=px1] {};
  \node [blue place,tokens=0] (checkx1) [xshift=-9mm,yshift=-13mm,label={[black]right:$P_x$},below of=incrCy] {};
  \node [purple place,tokens=0] (checkya) [xshift=9mm,yshift=-13mm,label={[black]left:$P_y$},below of=incrCy] {};
  \node [blue place,tokens=0] (okx) [xshift=0mm,yshift=-13mm,label={[black]right:$ok_x$},below of=checkx1] {};
  \node [purple place,tokens=0] (oky) [xshift=0mm,yshift=-13mm,label={[black]right:$ok_y$},below of=checkya] {};
  \node [normal place,tokens=0] (endCy) [xshift=0mm,yshift=-10mm,label={[black]below:$s_j$},below of=okx] {};

   \node [yellow transition] (emptyz0) [right of=pz0, label={[black]below:$[0,0]$}] {$\varepsilon_{z_{0}}$}
      edge [pre] (pz0)
      edge [post] (pz1);
   \node [normal transition] (tz1) [right of=pz1, label={[black]below:$[b,b]$}] {$t_{z=b}$}
      edge [pre] (pz1)
      edge [post] (fz1);
   \node [yellow transition] (emptyz1) [right of=fz1, label={[black]below:$[0,0]$}] {$\varepsilon_{z_{b}}$}
      edge [pre] (fz1)
      edge [post] (pza);
   \node [normal transition] (tza) [right of=pza,xshift=2mm, label={[black]below:$[a,a]$}] {$t_{z=a+b}$}
      edge [pre] (pza)
      edge [post] (fza);
   \node [yellow transition] (emptyza) [right of=fza,xshift=2mm, label={[black]below:$[0,0]$}] {$\varepsilon_{z_{a+b}}$}
      edge [pre] (fza);

    \node [normal transition] (txa) [below of=pxa, label={[black]left:$[a,a]$}] {$t_{x={a+b}}$}
      edge [pre] (pxa)
      edge [post] (fxa);
    \node [normal transition] (tx1) [below of=px1, label={[black]left:$[b,b]$}] {$t_{x=b}$}
      edge [pre] (px1)
      edge [post] (fx1);

    \node [normal transition] (tya) [below of=pya, label={[black]right:$[a,a]$}] {$t_{y=a+b}$}
      edge [pre] (pya)
      edge [post] (fya);
    \node [normal transition] (ty1) [below of=py1, label={[black]right:$[b,b]$}] {$t_{y=b}$}
      edge [pre] (py1)
     edge [post] (fy1);

          \node [yellow transition] (emptyxa) [below of=fxa, label={[black]left:$[0,0]$}] {$\varepsilon_{x_{a+b}}$}
      edge [pre] (fxa);
    \node [yellow transition] (emptyx1) [below of=fx1, label={[black]left:$[0,0]$}] {$\varepsilon_{x_b}$}
      edge [pre] (fx1)
      edge [post] (pxa);
    \node [yellow transition] (emptyya) [below of=fya, label={[black]right:$[0,0]$}] {$\varepsilon_{y_{a+b}}$}
      edge [pre] (fya);
    \node [yellow transition] (emptyy1) [below of=fy1, label={[black]right:$[0,0]$}] {$\varepsilon_{y_b}$}
      edge [pre] (fy1)
      edge [post] (pya);

    \node [green transition] (startIncrCy) [below of=incrCy,xshift=0mm,yshift=0mm, label={[blue]right:$[0,0]$}] {$start$}
      edge [pre] (incrCy)
      edge [purple,post] (checkya)
      edge [blue,post] (checkx1);

   \node [purple transition] (ya) [below of=checkya, label={[purple]right:$[0,0]$}] {$R(y)$}
      edge [purple,pre] (fy1)
      edge [purple,pre] (checkya)
      edge [purple,post] (oky)
      edge [purple,post] (py1);

  \node [blue transition] (x1) [below of=checkx1, label={[black]left:$[0,0]$}] {$R(x)$}
      edge [blue,pre] (fx1)
      edge [blue,pre] (checkx1)
      edge [blue,post] (okx)
      edge [blue,post] (px1);

    \node [green transition] (doneCy) [below of=okx, label={[blue]left:$[0,0]$}] {$zero$}
      edge [blue,pre] (okx)
      edge [pre, bend left=5] (fx1)
      edge [post, bend right=5] (fx1)
      edge [purple,pre] (oky)
      edge [pre,bend left=25] (fz1)
      edge [post] (endCy)
      edge [post,bend left=15] (pz0);
    \node [green transition] (doneCyp) [below of=oky, label={[blue]right:$[0,0]$}] {$not zero$}
      edge [blue,pre] (okx)
      edge [purple,pre] (oky);
\end{tikzpicture}
}\vspace*{-3mm}
\caption{Encoding 0-test over bounded-time: when in state $s_i$, if  $C_x=0$ then go to $s_j$}
\label{BD:zero-test}\vspace*{-6mm}
\end{figure}

\section{Symbolic semi-algorithms for parameter synthesis}
\label{sec:semialgo}

\subsection{State classes}

We now introduce the notion of state classes for pcTPNs. It was originally introduced for time Petri nets in~\cite{berthomieu-IFIP-83,berthomieu-TSE-91}, and extended for timing parameters in~\cite{traonouez-JUCS-09}, and for costs in~\cite{boucheneb-FORMATS-17}. We show that those two extensions seamlessly blend together.

For an  arbitrary sequence of transitions $\sigma = t_1\dots t_n \in T^*$, let $C_\sigma$ be the set of all
states that can be reached by the sequence $\sigma$ from any initial state $q_0$ : $C_\sigma =\{q \in Q | q_0 \xhookrightarrow{t_1}q_1\cdots \xhookrightarrow{t_n} q\}$.
 All the states of $C_\sigma$  share the same marking and can therefore be written as a pair $(m, D)$ where $m$ is the common marking and, if we note $\en{m}=\{t_1,\ldots,t_n\}$, then $D$ is the set of points $(\theta_{1}, \ldots,\theta_{n},c,v)$ such that $(m,I,c,v)\in C_\sigma$ and for all $t_i\in\en{m}$, $\theta_{i}\in I(t_i)$. For short, we will often write $(\vec{\theta},c,v)$ for such a point, with $\vec{\theta}=(\theta_1,\ldots,\theta_n)$ and a small abuse of notation. We denote by $\Theta$ the set of $\theta_i$ variables, of which we have one per transition of the net: for the sake of simplicity, we will usually use the same index to denote for instance that $\theta_i$ corresponds to transition $t_i$.

$C_\sigma$ is called a \emph{state class} and $D$ is its \emph{firing domain}.

Lemma~\ref{lemma:runs} equivalently characterises state classes, as a straightforward reformulation of the definition:
\begin{lemma}
    For all classes $C_\sigma=(m,D)$, $(\vec{\theta},c,v) \in D$ if and only if there exists a run $\rho$ in $v(\Net)$, and $I:\en{m}\rightarrow \Intervals{\Qp}$, such that $\sequence(\rho)=\sigma$, $(m,I,c)=\last(\rho)$, and $\vec{\theta}\in I$.
    \label{lemma:runs}
\end{lemma}

From Lemma~\ref{lemma:runs}, we can then deduce a characterisation of the ``next'' class, obtained by firing a firable transition from some other class. This is expressed by Lemma~\ref{lemma:succ}.
\begin{lemma}
Let $C_{\sigma} = (m,D)$ and $C_{\sigma.t_f} = (m',D')$, we have:
$$(\vec{\theta'}, c', v)\in D'\text{ iff }\exists(\vec{\theta},c, v)\in D\text{ s.t. }
\left\{\begin{array}{l}
    \forall t_i\in\en{m}, \theta_i - \theta_f\geq 0\\
    \forall t_i\in\en{m-\pre{t_f}}, \theta'_i=\theta_i-\theta_f\\
    \forall t_i\in\newen{m,t_f}, \theta'_i\in v(I_s)(t_i)\\
    c' = c + \costrate(m)*\theta_f + \discreteCost(t_f)
\end{array}\right.$$
    \label{lemma:succ}
\end{lemma}

\begin{proof}
    Consider $(\vec{\theta}', c', v)\in D'$. Then by Lemma~\ref{lemma:runs}, there exists a run $\rho'$ in $v(\Net)$, and $I':\en{m}\rightarrow \Intervals{\Qp}$, such that $\sequence(\rho')=\sigma.t_f$, $(m',I',c')=\last(\rho')$, and $\vec{\theta'}\in I'$.
    Consider the prefix $\rho$ of $\rho'$ such that $\sequence(\rho)=\sigma$. The last state of $\rho$ can be written $(m,I,c,v)$ for some $I$ and $c$. We know that $t_f$ is fired from $(m,I,c,v)$ so there exists some delay $d$ such that $\eft{I(t_f)}\leq d$ and for all other transitions $t_i$ enabled by $m$, $\lft{I(t_i)} \geq d$. Furthermore, $c = c'-\costrate(m)*d - \discreteCost(t_f)$.  It follows that there exists a point $\vec{\theta}\in I$ with the desired properties.

    The other direction is similar.
\end{proof}

Note that according to Lemma~\ref{lemma:succ}, $D'$ is not empty if and only if there exists $(\vec{\theta}, c, v)$ in $D$ such that for all $t_i\in\en{m}$, $\theta_i \geq \theta_f$. In that case we say that $t_f$ is \emph{firable} from $(m,D)$ and note $t_f \in\firable{(m,D)}$.

From Lemma~\ref{lemma:succ}, it follows that $C_{\sigma.t_f}$ can be computed from $C_{\sigma}$ using Algorithm~\ref{algo:nextclass}. Note that it is formally the same algorithm as in~\cite{boucheneb-FORMATS-17}.

Given a class $C$ and a transition $t$ firable from $C$, we note $\Next(C,t)$ the result of applying Algorithm~\ref{algo:nextclass} to $C$ and $t$.\medskip
\begin{algorithm}[!h]
    \begin{algorithmic}[1]
        \STATE $m'\leftarrow m -\pre{t_f}+\post{t_f}$
        \STATE $D'\leftarrow D\wedge \bigwedge_{i\neq f, t_i\in\en{m}} \theta_f \leq \theta_i$\label{line:fireFirst}
        \STATE for all $t_i\in\en{m-\pre{t_f}},i\neq f$, add variable $\theta'_i$ to $D'$, constrained by $\theta_i = \theta'_i + \theta_f$\label{line:renameThetas}
        \STATE add variable $c'$ to $D'$, constrained by $c'=c+\theta_f *\costrate(m)+\discreteCost(t_f)$\label{line:cost1}
        \STATE eliminate (by projection) variables $c$, $\theta_i$ for all $i$ from $D'$\label{line:cost2}
        \STATE for all $t_j\in\newen{m,t_f}$, add variable $\theta'_j$ to $D'$, constrained by $\theta'_j\in I_s(t_j)$\label{line:newTransitions}
    \end{algorithmic}
    \caption{Successor $(m', D')$ of $(m, D)$ by firing $t_f$}
    \label{algo:nextclass}
\end{algorithm}

Let $C_0=(m_0,D_0)$ be the initial class. Domain $D_0$ is defined by the constraints $\forall t_i\in\en{m_0}, \theta_i\in I_s(t_i)$, $\forall t\in T, I_s(t)\neq\emptyset$, and $c=0$.
This gives a convex polyhedron of $\Rp^{|\en{m_0}|+|\Params|+1}$; since all the operations on domains in Algorithm~\ref{algo:nextclass} are polyhedral, all the domains of state classes are also convex polyhedra.
Note that only enabled transitions are constrained in the domain of a state class.

\medskip
Naturally, we define the \emph{cost} of state class $C_\sigma$ as $\costclass(C_\sigma)=\costclass(\sigma)$.

\begin{example}
    For the net in Figure~\ref{fig:pctpn}, the initial state class is defined by $m_0=\{p_0,p_1\}$ and $D_0=\{\theta_0=a,2\leq \theta_1\leq 5, c = 0, a\geq 0\}$. Firing $t_0$ from $C_0$, we get $C_1=(m_1, D_1)$ with $m_1=m_0$, and $D_1=\{t_0=a, 0\leq t_1, 2-a\leq t_1 \leq 5-a, c=2+3a, a\leq 5\}$. More generally, after firing $n$ times $t_0$ consecutively, we get $C_n=(m_0, D_n)$ with $D_n=\{t_0=a, 0\leq t_1, 2-na\leq t_1 \leq 5-na, c=n(2+3a), na\leq 5\}$. Finally, when firing $t_1$ from $C_n$, with $n>0$, we get $C'_n=(\{p_0,p_2\}, D'_n)$ with $D'_n=\{0\leq t_0\leq a, (n+1)a -5 \leq t_0\leq (n+1)a -2, c= n(2+3a)+3(a-t_0), \frac{2}{n+1}\leq a\leq \frac{5}{n}\}$. When firing $t_1$ from $C_0$, we get $C'_0=(\{p_0,p_2\}, D'_0)$ with $D'_0 = \{0\leq t_0\leq a, a -5 \leq t_0\leq a -2, c= 3(a-t_0), 2\leq a\}$.
    \label{ex:sc}
\end{example}

In the next two subsections we present parameter synthesis algorithms based on the state classes with parameters and cost we have just defined.

\subsection{Bounded-cost synthesis semi-algorithm}

\begin{algorithm}[!h]
    \begin{algorithmic}[1]
    \STATE $\aPoly\leftarrow \emptyset$
    \STATE $\aPassed\leftarrow \emptyset$
    \STATE $\aWaiting\leftarrow \{(m_0, D_0)\}$
    \WHILE {$\aWaiting\neq \emptyset$}
        \STATE select $C_\sigma = (m, D)$ from $\aWaiting$\label{line:select}
        \IF {$m\in \aGoal$}\label{line:mInGoal}
            \STATE $\aPoly\leftarrow \aPoly\cup\proj{\big(D\cap (c \leq \cmax)\big)}{\Params}$\label{line:updatePoly}
        \ENDIF
        \IF {for all $C'\in \aPassed, C_\sigma\not\costincl C'$}\label{line:passed}
            \STATE add $C_\sigma$ to $\aPassed$\label{line:addToPassed}
            \STATE for all $t\in\firable{C_\sigma}$, add $C_{\sigma.t}$ to $\aWaiting$\label{line:addToWaiting}
        \ENDIF
    \ENDWHILE
        \RETURN $\aPoly$
\end{algorithmic}
\caption{Symbolic semi-algorithm computing all parameter valuations such that some markings are reachable with a bounded cost.}
\label{algo:boundedsynth}
\end{algorithm}

We start with the problem of finding parameter valuations for which we can reach some given markings, with a cost that is less or equal to a given constant $\cmax$.
In Algorithm~\ref{algo:boundedsynth}, we explore the symbolic state-space in a classical manner. Whenever a goal marking is encountered we collect the parameter valuations that allowed that marking to be reached with a cost less or equal to $\cmax$.

\medskip
The $\aPassed$ list records the visited symbolic states. Instead of checking new symbolic states for membership, we test a weaker relation denoted by $\costincl$: does there exist a visited state allowing more behaviors with a cheaper cost?

\medskip
For any state class $C = (m, D)$ and any point $(\vec{\theta},v)\in \proj{D}{\Theta\cup\Params}$, the optimal cost of $(\vec{\theta},v)$ in $D$ is defined by $\costclass_D(\vec{\theta},v) = \inf_{(\vec{\theta},c,v)\in D} c$.

\begin{definition}
    Let $C=(m, D)$ and $C'=(m', D')$ be two parametric cost state classes.
    We say that $C$ is subsumed by $C'$, which we denote by $C\costincl C'$ iff
        $m=m'$,
        $\proj{D}{\Theta\cup\Params}\subseteq \proj{D'}{\Theta\cup\Params}$,
        and for all $(\vec{\theta},v)\in \proj{D}{\Theta\cup\Params}, \costclass_{D'}(\vec{\theta},v) \leq \costclass_D(\vec{\theta},v)$.
    \label{def:costincl}
\end{definition}

The following result is a fairly direct consequence of Definition~\ref{def:costincl}:
\begin{lemma}
    Let $C_{\sigma_1}$ and $C_{\sigma_2}$ be two state classes such that $C_{\sigma_1}\costincl C_{\sigma_2}$.

    If a transition sequence $\sigma$ is firable from $C_{\sigma_1}$, it is also firable from $C_{\sigma_2}$ and $\costclass(C_{\sigma_1.\sigma}) \geq \costclass(C_{\sigma_2.\sigma})$.
    \label{lemma:costincl}
\end{lemma}

\begin{proof}
    Let $C_{\sigma_1}=(m_1,D_1)$ and $C_{\sigma_2}=(m_2,D_2)$.
    From Definition~\ref{def:costincl}, for any point $(\vec{\theta},c_1,v)\in D_1$, there exists a point $(\vec{\theta},c_2,v)\in D_2$ such that $c_2\leq c_1$. This implies that:
            (i) $\costclass(C_{\sigma_1}) \geq \costclass(C_{\sigma_2})$;
            (ii) if transition $t$ is firable from $C_{\sigma_1}$, then it is firable from $C_{\sigma_2}$ and $\Next(C_{\sigma_1},t) \costincl\Next(C_{\sigma_2},t)$.
    And the result follows by a straightforward induction.
\end{proof}

While $\costincl$ can be checked using standard linear algebra techniques, we can also reduce it to standard inclusion on polyhedra by removing the upper bounds on cost (an operation called cost relaxation)~\cite{boucheneb-FORMATS-17}.

\begin{example}
    Using the classes given in Example~\ref{ex:sc}, we see that for all $n$, for all $m\leq n$, we have $C_n \not\costincl C_m$ because $\proj{D_n}{\Theta\cup\Params}\not\subseteq \proj{D_m}{\Theta\cup\Params}$. This means Algorithm~\ref{algo:boundedsynth} will not terminate, whatever the value of the cost bound.

    Nonetheless, for a cost bound of $5$, we see that from $D'_0$ we obtain an empty set of parameters meeting the bound which is empty, as expected, because $t_0\leq a - 2$ so $c=3(a-t_0)$ is at least $6$. And for a bound greater of equal to $6$ we get $a\geq 2$. With a bound of $8$, we need to explore $C'_1$ to find that $a$ can actually be greater or equal to $1$, and so on.
\end{example}

Lemma~\ref{lemma:invbs} states a technical invariant of the while loop.
\begin{lemma}
    The following invariant holds after each iteration of the while loop in Algorithm~\ref{algo:boundedsynth}: for all $C_\sigma=(m,D)\in\aPassed$,
    \begin{enumerate}
    \itemsep=0.9pt
        \item for all prefixes $\sigma'$ of $\sigma$, $C_{\sigma'}\in\aPassed$;
        \item if $m\in\aGoal$ then $\proj{\big(D\cap (c \leq \cmax)\big)}{\Params} \subseteq \aPoly$;
        \item if $t$ is firable from $C_\sigma$:
            \begin{itemize}
            \itemsep=0.9pt
                \item either $C_{\sigma.t}\in\aWaiting$,
                \item or there exists $C'\in\aPassed$ such that $C_{\sigma.t}\costincl C'$.
            \end{itemize}
    \end{enumerate}
    \label{lemma:invbs}\vspace*{-2mm}
\end{lemma}

\begin{proof}
	We prove this lemma by induction.
	Before the while loop starts, $\aPassed$ is empty so the invariant is true.
	Let us now assume that the invariant holds for all iterations up to the $n$-th one, with $n\geq 0$, and that $\aWaiting\neq\emptyset$.
    Let $C_\sigma\in\aWaiting$ be the selected class at line~\ref{line:select}; to check whether the invariant still holds at the end of the $(n+1)$-th iteration, we only have to test the case where $C_\sigma$ is added to $\aPassed$ (which means that the condition at line~\ref{line:passed} is true).
	We can then check each part of the invariant:
	\begin{enumerate}
  \itemsep=0.9pt
		\item $C_\sigma$ was picked from $\aWaiting$ (line~\ref{line:select}); except for the initial class (for which $\sigma$ is empty, and therefore has no prefix), it means that, in a previous iteration, there was a sequence $\sigma'$ and a transition $t\in\firable{C_{\sigma'}}$ such that $\sigma=\sigma'.t$ (line~\ref{line:addToWaiting}) and $C_{\sigma'}\in\aPassed$ (line~\ref{line:addToPassed}).
            Since we add at most one state class to $\aPassed$ at each iteration, $C_{\sigma'}$ was added in a previous iteration and we can apply to it the induction hypothesis, which allows us to prove the first part of the invariant;
		\item lines~\ref{line:mInGoal} and~\ref{line:updatePoly} obviously imply the second part of the invariant;
        \item if $C_\sigma\in\aPassed$, then the condition of the if on line~\ref{line:passed} is true and then for any transition $t$ that is firable from $C_\sigma$, $C_{\sigma.t}$ is added to \aWaiting~(line~\ref{line:addToWaiting}) so the third part of the invariant holds for $C_\sigma$. Nevertheless, $C_\sigma$ itself is no longer in $\aWaiting$, and it is (except for the initial state class) the successor of some state class in $\aPassed$. But then we have only two possibilities: either $C_\sigma$ has been added to $\aPassed$ in line~\ref{line:addToPassed} if the condition on line~\ref{line:passed} was true, and certainly $C_\sigma\costincl C_\sigma$, or there exists $C'\in\aPassed$ such that $C_{\sigma}\costincl C'$ if that condition was false. Therefore the third part of the invariant holds.
	\end{enumerate}
	
	Both the basis case and the induction step are true: the result follows by induction.
\end{proof}

We can now prove the correctness of semi-algorithm, and its completeness when it terminates.
\begin{proposition}
    After any iteration of the while loop in Algorithm~\ref{algo:boundedsynth}:
    \begin{enumerate}
      \itemsep=0.9pt
        \item if $v\in\aPoly$, then there exists a run $\rho$ in $v(\Net)$ such that $\costrun(\rho)\leq \cmax$ and $\lastm(\rho)\in \aGoal$.
        \item if $\aWaiting = \emptyset$ then, for all parameter valuations $v$ such that there exists a run $\rho$ in $v(\Net)$ such that $\costrun(\rho)\leq \cmax$ and $\lastm(\rho)\in \aGoal$, we have $v\in\aPoly$.
    \end{enumerate}
    \label{prop:boundedcorrect}
\end{proposition}

\begin{proof}

\vspace*{-8mm}
    \begin{enumerate}
      \itemsep=0.9pt
        \item By induction on the while loop: initially, $\aPoly$ is empty so the result holds trivially. Suppose it holds after some iteration $n$, and consider iteration $n+1$.
            Let $v\in\aPoly$ after iteration $n+1$. If $v$ was already in $\aPoly$ after iteration $n$ then we can apply the induction hypothesis. Otherwise it means that if $C_\sigma=(m,D)$ is the class examined at iteration $n+1$, then $m\in\aGoal$ and $v\in \proj{\big(D\cap (c \leq \cmax)\big)}{\Params}$. This means that there exists some point $(\vec{\theta},c,v)\in D$ with $c\leq \cmax$. By Lemma~\ref{lemma:runs}, this means that there exists a run $\rho$ such that $(m,I,c,v)=\last(\rho)$, for some $I$ such that $\vec{\theta}\in I$, and therefore $\lastm(\rho)\in\aGoal$ and $\cost(\rho)\leq \cmax$.
        \item Let $v$ be a parameter valuation such that there exists a run $\rho$ in $v(\Net)$ such that $\cost(\rho)\leq \cmax$ and $\lastm(\rho)\in\aGoal$. Let $\sigma=\sequence(\rho)$. We proceed by induction on the length $n$ of the biggest suffix $\sigma_2$ of $\sigma$ such that, either $\sigma_2$ is empty or, if we note $\sigma=\sigma_1\sigma_2$, with the first element of $\sigma_2$ being transition $t$, then $C_{\sigma_1t}\not\in\aPassed$.

            If $n=0$, then $C_\sigma=(m,D)\in\aPassed$. By Lemma~\ref{lemma:runs}, $v\in\proj{D}{\Params}$ and $m\in \aGoal$. From the latter, with Lemma~\ref{lemma:invbs}, we have $\proj{\big(D\cap (c \leq \cmax)\big)}{\Params} \subseteq \aPoly$ and therefore $v\in\aPoly$ because $v\in \proj{\big(D\cap (c \leq \cmax)\big)}{\Params}$.

            Consider now $n>0$ and assume the property holds for $n-1$. Since $n>0$, then there exists a transition $t$ and a sequence $\sigma_3$ such that $\sigma_2=t.\sigma_3$. By definition of $\sigma_2$, we have $C_{\sigma_1}\in\aPassed$ but $C_{\sigma_1.t}\not\in\aPassed$. By Lemma~\ref{lemma:invbs}, since $\aWaiting=\emptyset$, there must exists some class $C_{\sigma'}\in\aPassed$ such that $C_{\sigma_1.t}\costincl C_{\sigma'}$. From Lemma~\ref{lemma:costincl}, sequence $\sigma_3$ is also firable from $C_{\sigma'}$ and $C_{\sigma'.\sigma_3}=(m,D')$, with $\cost(C_{\sigma'.\sigma_3}) \leq \cost(C_\sigma)\leq \cmax$. By Lemma~\ref{lemma:runs}, there exists thus a run $\rho'$ in $v(\Net)$, with $\sequence(\rho')=\sigma'.\sigma_3$, $\lastm(\rho')\in\aGoal$ and $\costrun(\rho')\leq \cmax$. Also, from Lemma~\ref{lemma:invbs} (item~1), we know that for all prefixes of $\sigma'$, the corresponding state class is in $\aPassed$, so the biggest suffix of $\sigma'.\sigma_3$ as defined above in the induction hypothesis has length less or equal to $n-1$, and the induction hypothesis applies to $\rho'$, which allows to conclude.
    \end{enumerate}
\end{proof}

In particular, if the algorithm terminates, then the waiting list is empty and $\aPoly$ is exactly the solution to the synthesis problem.

\subsection{Infimum-cost synthesis semi-algorithm}

We now address the optimal cost synthesis problem, that is compute the infimum cost over all runs (and thus over all parameter valuations) and all parameter valuations for which this optimal cost can be achieved.

\begin{algorithm}[!h]
    \begin{algorithmic}[1]
    \STATE $\aCost\leftarrow \infty$
    \STATE $\aPoly\leftarrow \emptyset$
    \STATE $\aPassed\leftarrow \emptyset$
    \STATE $\aWaiting\leftarrow \{(m_0, D_0)\}$
    \WHILE {$\aWaiting\neq \emptyset$}
        \STATE select $C_\sigma = (m, D)$ from $\aWaiting$
        \IF {$m\in \aGoal$}\label{line:inf_startif}
            \IF {$\costclass(C_\sigma) < \aCost$}
                \STATE $\aCost\leftarrow \costclass(C_\sigma)$
                \STATE $\aPoly\leftarrow \proj{\big(D\cap (c = \aCost)\big)}{\Params}$
            \ELSIF {$\costclass(C_\sigma) = \aCost$}
                \STATE $\aPoly\leftarrow \aPoly\cup \proj{\big(D\cap (c = \aCost)\big)}{\Params}$
            \ENDIF
        \ENDIF\label{line:inf_endif}
        \IF {for all $C'\in \aPassed, C_\sigma\not\costincl C'$}\label{line:inf_passed}
            \STATE add $C_\sigma$ to $\aPassed$
            \STATE for all $t\in\firable{C_\sigma}$, add $C_{\sigma.t}$ to $\aWaiting$
        \ENDIF
    \ENDWHILE
        \RETURN $(\aCost, \aPoly)$
\end{algorithmic}

\caption{Symbolic semi-algorithm computing all parameter valuations such that some markings are reachable with the infimum cost.}
\label{algo:infsynth}
\end{algorithm}

Algorithm~\ref{algo:infsynth} explores the state-space in the same way as Algorithm~\ref{algo:boundedsynth}. The difference lies in how it collects ``good'' parameter valuations: the idea is to always store the parameter valuations for which the minimal possible cost in the class is realisable.

\begin{example}
    As we have seen previously, the minimum cost for $C'_0$ is $6$ with $a\geq 2$ and this is actually optimal because for all $n>0$, the minimum cost for $C'_n$ will be $n(2+3a) +  3(a-t_0)$ with $t_0\leq (n+1)a -2$, i.e., $a-t_0\geq 2 - na$ and $c\geq 2n +3na +6 -3na=2n+6>6$.
\end{example}

As before we first prove a simple technical invariant on the while loop.
\begin{lemma}
    The following invariant holds after each iteration of the while loop in Algorithm~\ref{algo:infsynth}:
        for all $C_\sigma=(m,D)\in\aPassed$,
    \begin{enumerate}
        \item for all prefixes $\sigma'$ of $\sigma$, $C_{\sigma'}\in\aPassed$;
        \item if $t$ is firable from $C_\sigma$:
            \begin{itemize}
                \item either $C_{\sigma.t}\in\aWaiting$,
                \item or there exists $C'\in\aPassed$ such that $C_{\sigma.t}\costincl C'$.
            \end{itemize}
        \item $\cost(C_\sigma) \geq \aCost$;
        \item if $m\in \aGoal$ and $\costclass(C_\sigma) = \aCost$
            then
            $\proj{\big(D\cap (c = \costclass(C_\sigma))\big)}{\Params}\subseteq \aPoly$;
    \end{enumerate}
    \label{lemma:invinfs}
\end{lemma}

\begin{proof}
    The proof works by induction exactly as the one of Lemma~\ref{lemma:invbs}. The specific last two items are a direct consequence of the updates in lines \ref{line:inf_startif} to \ref{line:inf_endif} and of the induction hypothesis.
\end{proof}

We can now prove correctness and completeness (the latter provided the semi-algorithm terminates).
\begin{proposition}
    After any iteration of the while loop in Algorithm~\ref{algo:infsynth}:
    \begin{enumerate}

        \item if $v\in\aPoly$, then there exists a run $\rho$ in $v(\Net)$ such that $\costrun(\rho)=\aCost$ and $\lastm(\rho)\in \aGoal$.
        \item if $\aWaiting = \emptyset$ then
               if some marking in $\aGoal$ is reachable for some parameter valuation then $\aCost = \min_{\rho\in\Runs(\calN),\lastm(\rho)\in\aGoal} \cost(\rho)$ otherwise $\aCost=+\infty$.
        \item if $\aWaiting = \emptyset$ then
                for all parameter valuations $v$ such that there exists a run $\rho$ in $v(\Net)$ such that $\costrun(\rho) = \aCost$ and $\lastm(\rho)\in \aGoal$, we have $v\in\aPoly$;
    \end{enumerate}
    \label{prop:infcorrect}
\end{proposition}

\begin{proof}

\vspace*{-7mm}
    \begin{enumerate}
        \item This works as in the proof for Proposition~\ref{prop:boundedcorrect}, replacing $c\leq\cmax$ by $c=\aCost$.

                \item If no marking in $\aGoal$ is reachable then clearly, from its initialisation, $\aCost=+\infty$ when the algorithm terminates with $\aWaiting=\emptyset$.

            If some marking in $\aGoal$ is reachable for some parameter valuation $v$, let $\rho$ be one of the runs reaching $\aGoal$ with the smallest cost and let $\sigma=\sequence(\rho)$. First remark that $\cost(C_\sigma)=\cost(\rho)$ because $\rho$ realises the minimum cost of all runs reaching $\aGoal$, and therefore in particular of all runs firing transition sequence $\sigma$.

            We prove that $\aCost=\cost(\rho)$. We proceed again by induction on the length $n$ of the biggest suffix $\sigma_2$ of $\sigma$ such that, either $\sigma_2$ is empty or, if we note $\sigma=\sigma_1\sigma_2$, with the first element of $\sigma_2$ being transition $t$, then $C_{\sigma_1t}\not\in\aPassed$.

            If $n=0$, then $C_\sigma=(m,D)\in\aPassed$. Then, by Lemma~\ref{lemma:invinfs}, $\aCost\leq \cost(C_\sigma)$, which implies $\aCost<+\infty$. It is therefore clear that there is some $\sigma'=(m',D')\in\aPassed$ such that $\aCost=\cost(C_{\sigma'})$. If $\aCost< \cost(C_{\sigma})$, then by Lemma~\ref{lemma:runs} there exists a run $\rho'$ reaching $\aGoal$ with $\cost(\rho')<\cost(\rho)$, which is not possible since $\rho$ realises the minimum cost. Therefore $\aCost=\cost(C_\sigma)=\cost(\rho)$.

            Consider now $n>0$ and assume the property holds for $n-1$. Since $n>0$, then there exists a transition $t$ and a sequence $\sigma_3$ such that $\sigma_2=t.\sigma_3$. By definition of $\sigma_2$, we have $C_{\sigma_1}\in\aPassed$ but $C_{\sigma_1.t}\not\in\aPassed$. By Lemma~\ref{lemma:invbs}, since $\aWaiting=\emptyset$, there must exists some class $C_{\sigma'}$ such that $C_{\sigma_1.t}\costincl C_{\sigma'}$. From Lemma~\ref{lemma:costincl}, sequence $\sigma_3$ is also firable from $C_{\sigma'}$ and $C_{\sigma'.\sigma_3}=(m,D')$, with $\cost(C_{\sigma'.\sigma_3}) \leq \cost(C_\sigma)=\cost(\rho)$.  By Lemma~\ref{lemma:runs}, there exists thus a run $\rho'$, with $\sequence(\rho')=\sigma'.\sigma_3$, $\lastm(\rho')\in\aGoal$ and $\costrun(\rho')\leq \cost(\rho)$. Since $\rho$ realises the minimum of the cost for runs reaching $\aGoal$, this means that $\cost(\rho') = \cost(\rho)$. From Lemma~\ref{lemma:invinfs} (item~1),  we know that for all prefixes of $\sigma'$, the corresponding state class is in $\aPassed$, so the biggest suffix of $\sigma'.\sigma_3$ as defined above in the induction hypothesis has length less or equal to $n-1$, and the induction hypothesis applies to $\rho'$, which allows to conclude.

            \item Let $v$ be a parameter valuation such that there exists a run $\rho$ in $v(\Net)$ such that $\cost(\rho)=\aCost$ and $\lastm(\rho)\in\aGoal$. Let $\sigma=\sequence(\rho)$. We proceed by induction on the length $n$ of the biggest suffix $\sigma_2$ of $\sigma$ such that, either $\sigma_2$ is empty or, if we note $\sigma=\sigma_1\sigma_2$, with the first element of $\sigma_2$ being transition $t$, then $C_{\sigma_1t}\not\in\aPassed$.

            If $n=0$, then $C_\sigma=(m,D)\in\aPassed$. By Lemma~\ref{lemma:runs}, $v\in\proj{D}{\Params}$ and $m\in \aGoal$. From the latter, with Lemma~\ref{lemma:invinfs}, we have $\proj{\big(D\cap (c =\aCost)\big)}{\Params} \subseteq \aPoly$ and therefore $v\in\aPoly$ because $v\in \proj{\big(D\cap (c =\aCost)\big)}{\Params}$.

            Consider now $n>0$ and assume the property holds for $n-1$. Since $n>0$, then there exists a transition $t$ and a sequence $\sigma_3$ such that $\sigma_2=t.\sigma_3$. By definition of $\sigma_2$, we have $C_{\sigma_1}\in\aPassed$ but $C_{\sigma_1.t}\not\in\aPassed$. By Lemma~\ref{lemma:invbs}, since $\aWaiting=\emptyset$, there must exists some class $C_{\sigma'}\in\aPassed$ such that $C_{\sigma_1.t}\costincl C_{\sigma'}$. From Lemma~\ref{lemma:costincl}, sequence $\sigma_3$ is also firable from $C_{\sigma'}$ and $C_{\sigma'.\sigma_3}=(m,D')$, with $\cost(C_{\sigma'.\sigma_3}) \leq \cost(C_\sigma)=\aCost$. By Lemma~\ref{lemma:runs}, there exists thus a run $\rho'$ in $v(\Net)$, with $\sequence(\rho')=\sigma'.\sigma_3$, $\lastm(\rho')\in\aGoal$ and $\costrun(\rho')\leq \aCost$.
            Since $\aWaiting=\emptyset$, we can use item~2 above and therefore, since $\aGoal$ is reachable, $\aCost$ is finite and is the minimum cost of runs reaching $\aGoal$ (and in particular for those with parameter valuation $v$). So $\costrun(\rho')=\aCost$.
            Furthermore, from Lemma~\ref{lemma:invinfs} (item~1),  we know that for all prefixes of $\sigma'$, the corresponding state class is in $\aPassed$, so the biggest suffix of $\sigma'.\sigma_3$ as defined above in the induction hypothesis has length less or equal to $n-1$, and the induction hypothesis applies to $\rho'$, which allows to conclude.
    \end{enumerate}
\end{proof}

\section{Restricting to integer parameters}
\label{sec:integer}

Obviously, in general, (semi-)Algorithm~\ref{algo:boundedsynth} will not terminate, since the emptiness problem for the set it computes is undecidable. Semi-algorithm~\ref{algo:infsynth} will not terminate either for similar reasons.

\medskip
To ensure termination, we can however follow the methodology of \cite{jovanovic-TSE-15}: we require that parameters are bounded integers and, instead of just enumerating the possible parameter values, we propose a modification of the symbolic state computation to compute these integer parameters symbolically. For this we rely on the notion of integer hull.

\medskip
We call \emph{integer valuation} a $\Z$-valuation. Note that a $\Z$-valuation is also an $\R$-valuation, and
given a set $D$ of $\R$-valuations, we denote by $\IV(D)$ the set of integer valuations in $D$.

The \emph{convex hull} of a set $D$ of valuations, denoted by $\Conv(D)$, is the intersection of all the convex sets of valuations that contain $D$.

The \emph{integer hull} of a set $D$ of valuations, denoted by $\IH(D)$, is defined as the convex hull of the integer valuations in $D$: $\IH(D)=\Conv(\IV(D))$.

For a state class $C=(m,D)$, we write $\IH(C)$ for $(m,\IH(D))$.

\medskip
Before we see how our result can be adapted for the restriction to integer parameter valuations, and from there how we can enforce termination of the symbolic computations when parameters are assumed to be bounded, we need some results on the structure of the polyhedra representing firing domains of cost TPNs.

\medskip
By the Minkowski-Weyl theorem (see e.g. \cite{schrijver-book-86}), every convex polyhedron can be either described as a set of linear inequalities, as seen above, or by a set of \emph{generators}. More precisely, for the latter: if $d$ is the dimension of polyhedron $\mathcal{P}$, there exists $v_1,\ldots,v_p, r_1,\ldots,r_s\in \R^{d}$, such that for all points $x\in \mathcal{P}$, there exists $\lambda_1, \ldots, \lambda_p\in\R, \mu_1, \ldots,\mu_s\in\Rp$ such that $\sum_i\lambda_i = 1$ and $x=\sum_i\lambda_iv_i + \sum_i\mu_ir_i$.
    The $v_i$'s are called the \emph{vertices} of $\mathcal{P}$ and the $r_i$'s are the \emph{extremal rays} of $\mathcal{P}$. The latter correspond to the directions in which the polyhedron is infinite. In our case, they correspond to transitions with a (right-)infinite static interval, and possibly the cost.

\medskip
    A classical property of vertices, which can also be used as a definition, is as follows: $\vec{v}$ is a vertex of $\mathcal{P}$ iff for all non-null vectors $\vec{x}\in\R^d$, either $\vec{v}+\vec{x}\not\in \mathcal{P}$ or $\vec{v}-\vec{x}\not\in \mathcal{P}$ (or both), $+$ and $-$ being understood component-wise.

\begin{proposition}
    Let $\Net$ be a (non-parametric) cost TPN and let $C=(m,D)$ be one of its state classes, then $D$ has integer vertices.
    \label{prop:fullintv}
\end{proposition}

\begin{proof}
    We have proved in \cite{boucheneb-FORMATS-17} that the domain $D$ of a state class of a cost TPNs, with removed upper bounds on cost (so-called relaxed classes), can be partitioned into a union of simpler polyhedra $\bigcup_{i=1}^n D_i$ that have the following key properties: (1) by projecting the cost out we obtain a convex polyhedron $\proj{D_i}{\Theta}$ with integer vertices (actually a \emph{zone}, as in~\cite{larsen-FCT-95,berthomieu-TSE-91}), and (2) these simpler polyhedra all have exactly one constraint on the cost variable, i.e., of the form $c\geq \ell(\vec{\theta})$, with integer coefficients. Note that the same result can be obtained, with the same technique, if we consider non-relaxed state classes, except that, we also have an upper bound on cost that is always greater or equal to the lower bound. We prove in Lemma~\ref{lemma:intv} that each of these simpler polyhedra also has integer vertices. Since $D$ and each of the $D_i$'s are convex and since $D=\bigcup_i D_i$, $D$ is equal to the convex hull of the vertices of the $D_i$'s and therefore $D$ also has integer vertices.
\end{proof}

   \begin{lemma}
       Let $D$ be a convex polyhedron on variables $\theta_1,\ldots,\theta_n,c$ such that the projection of $D$ on the $\theta$ variables has integer vertices, and there are two constraints on $c$ of the form $c\geq \ell(\theta_1, \ldots, \theta_n)$ and $c\leq \ell'(\theta_1, \ldots, \theta_n)$, with $\ell$ and $\ell'$ linear terms with integer coefficients, such that $\ell(\theta_1, \ldots, \theta_n) \leq \ell'(\theta_1, \ldots, \theta_n)$, for all values of the $\theta_i$'s.

        Then, the vertices of $D$ are the points $(\theta_1,\ldots,\theta_n,\ell(\theta_1,\ldots,\theta_n))$ and $(\theta_1,\ldots,\theta_n,$ $\ell'(\theta_1,\ldots,\theta_n))$ such that $(\theta_1,\ldots,\theta_n)$ is a vertex of $\proj{D}{\Theta}$, and they are integer points.
        \label{lemma:intv}
    \end{lemma}

    \begin{proof}
        Recall here that we consider all constraints in $D$ to be non-strict so all polyhedra are topologically closed. The reasoning extends with no difficulty to non-necessarily-closed polyhedra by considering so-called \emph{closure points} in addition to vertices~\cite{bagnara-FAC-05}.

        Consider a non-vertex point $\vec{\theta}$ in $\proj{D}{\Theta}$ and let $(\vec{\theta}, c)$ be a point of $D$. Then using the form of the unique cost constraint, we have $c\geq \ell(\vec{\theta})$.
        Now since $\vec{\theta}$ is not a vertex, there exists a vector $\vec{x}$ such that both $\vec{\theta}+\vec{x}$ and $\vec{\theta}-\vec{x}$ belong to $\proj{D}{\Theta}$.  Then, for sure, $(\vec{\theta}+\vec{x}, \ell(\vec{\theta}+\vec{x}))\in D$ and $(\vec{\theta}-\vec{x}, \ell(\vec{\theta}-\vec{x}))\in D$. And since $\ell$ is linear, $(\vec{\theta}+\vec{x}, \ell(\vec{\theta})+\ell(\vec{x})))\in D$, i.e., $(\vec{\theta},\ell(\vec{\theta})) + (\vec{x}, \ell(\vec{x}))\in D$. And similarly, $(\vec{\theta},\ell(\vec{\theta})) - (\vec{x}, \ell(\vec{x}))\in D$.
        Using again the form of the unique cost constraint, and the fact that $c\geq \ell(\vec{\theta})$, we finally have $(\vec{\theta}, c) + (\vec{x}, \ell(\vec{x}))\in D$ and $(\vec{\theta}, c) - (\vec{x}, \ell(\vec{x}))\in D$, that is, $(\vec{\theta}, c)$ is not a vertex of $D$.

        By contraposition, any vertex of $D$ extends a vertex of $\proj{D}{\Theta}$, and using a last time the form of the cost constraint, any vertex of $D$, is of the form $(\vec{\theta}, \ell(\vec{\theta}))$, with $\vec{\theta}$ a vertex of $\proj{D}{\Theta}$: suppose $(\vec{\theta}, c)$ is a vertex of $D$, with $c>\ell(\vec{\theta})$, then for $\vec{x}$ defined with $c - \ell(\vec{\theta})$ on the cost variable, and $0$ on all other dimensions, we clearly have both $(\vec{\theta}, c) + \vec{x}$ and $(\vec{\theta}, c) - \vec{x}$ in $D$, which is a contradiction.

        We conclude by remarking that, since $\proj{D}{\Theta}$ has integer vertices, all the coordinates of $\vec{\theta}$ are integers, and since $\ell$ has integer coefficients then $\ell(\vec{\theta})$ is an integer.

        We can deal with the upper bound defined by $\ell'$ in exactly the same way.
    \end{proof}

From Proposition~\ref{prop:fullintv}, we can prove the following lemma that will be very useful in the subsequent proofs.
\begin{lemma}
    Let $(m, D)$ be a state class of a pcTPN and let $(\vec{\theta}, c, v)$ be a point in $D$.

    If $v$ is an integer valuation, then $(\vec{\theta}, c, v)\in\IH(D)$.
    \label{lemma:intpoint}
\end{lemma}

\begin{proof}
    Since $(\vec{\theta}, c, v)\in D$ then $(\vec{\theta}, c)\in v(D)$.
    By Proposition~\ref{prop:fullintv}, $v(D)$ being the firing domain of a state class in a (non-parametric) cost TPN, it has integer vertices, and therefore $v(D)= \IH(v(D))$. Point $(\vec{\theta}, c)$ is therefore a convex combination of integer points in $v(D)$. Clearly, for all integer points $(\vec{\theta'},c ')$ in $v(D)$, we have that $(\vec{\theta'}, c', v)$ is an integer point of $D$. Since $D$ is convex, this implies that $(\vec{\theta}, c, v)\in \IH(D)$.
\end{proof}

When we restrict ourselves to integer parameter but continue to work symbolically, we need to adjust the definitions of the firability of a transition from a class and of the cost of a class.

First, a transition $t_f$ is firable for integer parameter valuations from a class $(m, D)$, call this \intfirable{}, if there exists an \emph{integer} parameter valuation $v$ and a point $(\vec{\theta}, c, v)$ in $D$ such that for all transitions $t_i\in\en{m}, \theta_i\geq \theta_f$.

\begin{lemma}
    Let $C=(m,D)$ be a state class.
    Transition $t_f\in\en{m}$ is \intfirable{} from $C$ if and only if it is firable (not necessarily \intfirable) from $(m,\IH(D))$.
    \label{lemma:intfirable}
\end{lemma}

\begin{proof}
    $\Leftarrow$: trivial because $\IH(D)\subseteq D$.

    $\Rightarrow$: since $t_f$ is \intfirable{} from $C$, there exists an integer parameter valuation $v$, and $(\vec{\theta},c,v)\in D$ such that for all transitions $t_i\in\en{m}, \theta_i\geq \theta_f$. And the result follows from Lemma~\ref{lemma:intpoint} because $v$ is an integer valuation.
\end{proof}

Second, the cost of a class $C=(m,D)$, for integer parameters, is $\costclassint(C)=\inf_{(\vec{\theta},c,v)\in D, v\in\N^\Params} c$.

Lemma~\ref{lemma:intcost} is a direct consequence of Lemma~\ref{lemma:intpoint}:
\begin{lemma}
    Let $(m,D)$ be a state class. We have: $\costclassint((m,D))=\costclass((m,\IH(D))$.
    \label{lemma:intcost}
\end{lemma}

\begin{lemma}
    If $v$ is an integer parameter valuation, then for all classes $C_\sigma=(m,D)$, $(\vec{\theta},c,v) \in \IH(D)$ if and only if there exists a run $\rho$ in $v(\Net)$, and $I:\en{m}\rightarrow \Intervals{\Qp}$, such that $\sequence(\rho)=\sigma$, $(m,I,c)=\last(\rho)$, and $\vec{\theta}\in I$.
    \label{lemma:iruns}
\end{lemma}
\begin{proof}
    $\Rightarrow$: if $(\vec{\theta},c,v)\in\IH(D)$ then it is also in $D$ and the result follows from Lemma~\ref{lemma:runs}.

    $\Leftarrow$: by Lemma~\ref{lemma:runs}, we know that there exists some $(\vec{\theta},c,v)\in D$, and since $v$ is an integer valuation, by Lemma~\ref{lemma:intpoint}, $(\vec{\theta},c,v)\in \IH(D)$.
\end{proof}
\begin{lemma}
    Let $C_{\sigma_1}$ and $C_{\sigma_2}$ be two state classes such that $\IH(C_{\sigma_1})\costincl \IH(C_{\sigma_2})$.

    If a transition sequence $\sigma$ is \intfirable{} from $C_{\sigma_1}$ it is also \intfirable{} from $C_{\sigma_2}$ and $\costclassint(C_{\sigma_1.\sigma}) \geq \costclassint(C_{\sigma_2.\sigma})$.
    \label{lemma:intcostincl}
\end{lemma}
\begin{proof}
    Let $C_{\sigma_1}=(m_1,D_1)$ and $C_{\sigma_2}=(m_2,D_2)$.
    From Definition~\ref{def:costincl}, for any point $(\vec{\theta},c_1,v)\in \IH(D_1)$, there exists a point $(\vec{\theta},c_2,v)\in \IH(D_2)$ such that $c_2\leq c_1$. With Lemma~\ref{lemma:intfirable} and Lemma~\ref{lemma:intcost}, this implies that:
            (i) $\costclassint(C_{\sigma_1}) \geq \costclassint(C_{\sigma_2})$;
            (ii) if transition $t$ is \intfirable{} from $C_{\sigma_1}$, then it is \intfirable{} from $C_{\sigma_2}$ and $\Next(C_{\sigma_1},t) \costincl\Next(C_{\sigma_2},t)$.
    And, as before, the result follows by a straightforward induction.
\end{proof}

\begin{algorithm}[!h]
    \begin{algorithmic}[1]
    \STATE $\aPoly\leftarrow \emptyset$
    \STATE $\aPassed\leftarrow \emptyset$
    \STATE $\aWaiting\leftarrow \{(m_0, D_0)\}$
    \WHILE {$\aWaiting\neq \emptyset$}
        \STATE select $C_\sigma = (m, D)$ from $\aWaiting$
        \IF {$m\in \aGoal$}
            \STATE $\aPoly\leftarrow \aPoly\cup\proj{\big(\IH(D)\cap (c \leq \cmax)\big)}{\Params}$
        \ENDIF
        \IF {for all $C'\in \aPassed, \IH(C_\sigma)\not\costincl \IH(C')$}
            \STATE add $C_\sigma$ to $\aPassed$
            \STATE for all $t\in\firable{\IH(C_\sigma)}$, add $C_{\sigma.t}$ to $\aWaiting$
        \ENDIF
    \ENDWHILE
        \RETURN $\aPoly$
\end{algorithmic}
    \caption{Restriction of (semi-)Algorithm~\ref{algo:boundedsynth} to integer parameter valuations.}
\label{algo:intboundedsynth}
\end{algorithm}

\begin{example}
    Let us compute the integer hulls of some of the state classes of the net in Figure~\ref{fig:pctpn}.
    $D_0$ and $D_1$ already have integer vertices, but starting from $n>1$, we have in $C_n$ that $a\leq \frac{5}{n}$, and so $D_2$ does not have integer vertices. The integer hull of $D_2$ is $\{t_0=a, 0\leq t_1, 2-a\leq t_1\leq 5-a, c=2(2+3a), 0\leq a\leq 2\}$. For $n\in\{3,4,5\}$, we have $\IH(D_n)=\{t_0=a, 0\leq t_1, 2-2a\leq t_1\leq 5-na, c=n(2+3a), 0\leq a\leq 1\}$. And finally, for $n\geq 6$, $\IH(D_n)=\{t_0=0, 2\leq t_1\leq 5, c=2n, a=0\}$. So $C_7\costincl C_6$ because $D_7=D_6$ and, for the costs, $14\geq 12$. Note that we actually already had $C_2\costincl C_1$ because $D_2\subset D_1$ and on $D_2$, i.e., here for all $a\leq 2$, we have $2(2+3a) \geq 2+3a$.
\end{example}

Using Lemma~\ref{lemma:iruns} instead of Lemma~\ref{lemma:runs}, and Lemma~\ref{lemma:intcostincl} instead of Lemma~\ref{lemma:costincl} in the proof of Proposition~\ref{prop:boundedcorrect}, we get the following proposition, stating the completeness and soundness of Algorithm~\ref{algo:intboundedsynth}.

\begin{proposition}
    After any iteration of the while loop in Algorithm~\ref{algo:intboundedsynth}:
    \begin{enumerate}
    \itemsep=0.9pt
        \item if $v\in\aPoly$ and $v$ is an integer parameter valuation then there exists a run $\rho$ in $v(\Net)$ such that $\costrun(\rho)\leq \cmax$ and $\lastm(\rho)\in \aGoal$.
        \item if $\aWaiting = \emptyset$ then for all integer parameter valuations $v$ such that there exists a run $\rho$ in $v(\Net)$ such that $\costrun(\rho)\leq \cmax$ and $\lastm(\rho)\in \aGoal$, we have $v\in\aPoly$.
    \end{enumerate}
    \label{prop:intboundedcorrect}
\end{proposition}

\begin{algorithm}[!h]
    \begin{algorithmic}[1]
    \STATE $\aCost\leftarrow \infty$
    \STATE $\aPoly\leftarrow \emptyset$
    \STATE $\aPassed\leftarrow \emptyset$
    \STATE $\aWaiting\leftarrow \{(m_0, D_0)\}$
    \WHILE {$\aWaiting\neq \emptyset$}
        \STATE select $C_\sigma = (m, D)$ from $\aWaiting$
        \IF {$m\in \aGoal$}
            \IF {$\costclass(\IH(C_\sigma)) < \aCost$}
                \STATE $\aCost\leftarrow \costclass(\IH(C_\sigma))$
                \STATE $\aPoly\leftarrow \proj{\big(\IH(D)\cap (c = \aCost)\big)}{\Params}$
            \ELSIF {$\costclass(\IH(C_\sigma)) = \aCost$}
                \STATE $\aPoly\leftarrow \aPoly\cup \proj{\big(\IH(D)\cap (c = \aCost)\big)}{\Params}$
            \ENDIF
        \ENDIF
        \IF {for all $C'\in \aPassed, \IH(C_\sigma)\not\costincl \IH(C')$}
            \STATE add $C_\sigma$ to $\aPassed$
            \STATE for all $t\in\firable{\IH(C_\sigma)}$, add $C_{\sigma.t}$ to $\aWaiting$
        \ENDIF
    \ENDWHILE
        \RETURN $(\aCost, \aPoly)$
\end{algorithmic}

    \caption{Restriction of (semi-)Algorithm~\ref{algo:infsynth} to integer parameter valuations.}
\label{algo:intinfsynth}
\end{algorithm}

Similarly, we can prove the completeness and soundness of Algorithm~\ref{algo:intinfsynth}:
\begin{proposition}
    After any iteration of the while loop in Algorithm~\ref{algo:intinfsynth}:
    \begin{enumerate}
\itemsep=0.9pt
        \item if $v\in\aPoly$ and $v$ is an integer valuation, then there exists a run $\rho$ in $v(\Net)$ such that $\costrun(\rho)=\aCost$ and $\lastm(\rho)\in \aGoal$.
        \item if $\aWaiting = \emptyset$ then
               if some marking in $\aGoal$ is reachable for some integer parameter valuation then $\aCost = \min_{\rho\in\Runs(\calN),\lastm(\rho)\in\aGoal} \cost(\rho)$ otherwise $\aCost=+\infty$.
        \item if $\aWaiting = \emptyset$ then
                for all integer parameter valuations $v$ such that there exists a run $\rho$ in $v(\Net)$ such that $\costrun(\rho) = \aCost$ and $\lastm(\rho)\in \aGoal$, we have $v\in\aPoly$;
    \end{enumerate}
    \label{prop:intinfcorrect}
\end{proposition}

In Algorithms~\ref{algo:intboundedsynth} and \ref{algo:intinfsynth}, we compute state classes as usual then handle them via their integer hulls. We can actually simply integrate integer hulls at the end of Algorithm~\ref{algo:nextclass} and use Algorithm~\ref{algo:boundedsynth} with this updated successor computation as proved by Lemma~\ref{lemma:keepih}.

\begin{lemma}
    Let $(m, D)$ be a state class of a pcTPN $\Net$, and $t$ a transition firable from $C$. Let $(m',D') = \Next((m,D), t)$ and $(m'',D'')= \Next((m,\IH(D)), t)$. Then $m'' = m'$ and $\IH(D'') = \IH(D')$.
    \label{lemma:keepih}
\end{lemma}

\begin{proof}
    The equality of markings is trivial so we focus on firing domains.

    By definition of the integer hull, we have $\IH(D)\subseteq D$. Since the computation of the next class domain is non-decreasing with respect to inclusion, we then have $D''\subseteq D'$. Taking the integer hull is also non-decreasing wrt. inclusion, so $\IH(D'')\subseteq \IH(D')$.

    Consider now an integer point $(\vec{\theta'},c',v)$ in $D'$. Then $(\vec{\theta'},c')\in v(D')$. Consider state class computations in the (non-parametric) cost TPN $v(\Net)$: there exists some point $(\vec{\theta},c)$ in $v(D)$ such that $(m', \vec{\theta'}, c')\in\Next((m,\{(\vec{\theta},c)\}), t)$. Since $(\vec{\theta}, c, v)$ thus belongs to $D$ and since $v$ is an integer parameter valuation, by Lemma~\ref{lemma:intpoint}, we have that $(\vec{\theta}, c, v)\in \IH(D)$. Thus $(\vec{\theta'}, c', v)\in D''$ and since it is an integer point, it is in $\IH(D'')$.
\end{proof}

\section{Termination of Algorithms~\ref{algo:intboundedsynth} and \ref{algo:intinfsynth}}
\label{sec:termination}

We now consider that parameter valuations are bounded by some value $M_1\in\N$ (and that they still have integer values).
We also assume that, for all integer parameter valuations, there exists $M_2\in\Z$ such that for all runs $\rho$ in $v(\Net)$, $\costrun(\rho) \geq M_2$: this allows us, as in \cite{boucheneb-FORMATS-17,bouyer-CAV-16}, to keep Algorithms~\ref{algo:intboundedsynth} and \ref{algo:intinfsynth} simple by doing away with negative cost loop-checking.
Finally, we assume the net itself is bounded: there exists $M_3\in\N$ such that for all reachable markings $m$, for all places $p$, $m(p)\leq M_3$.

To prove the termination of Algorithm~\ref{algo:intboundedsynth} and \ref{algo:intinfsynth} under these assumptions, we consider $\costrincl$ the symmetric relation to $\costincl$, such that $x\costrincl y$ iff $y\costincl x$. We prove that it is a well quasi-order (wqo), i.e., that for every infinite sequence of state classes, there exist $C$ and $C'$ in the sequence, with $C$ strictly preceding $C'$ such that $C\costrincl C'$. This implies that the exploration of children in Algorithm~\ref{algo:intboundedsynth} and \ref{algo:intinfsynth} will always eventually stop.

\begin{proposition}
    Let $\Net$ be a bounded pcTPN, with bounded integer parameters and such that the cost of all runs is uniformly lower-bounded for all integer parameter valuations.

    Relation $\costrincl$ is well-quasiorder on the set of state classes of $\Net$.
    \label{prop:wqo}
\end{proposition}

\begin{proof}
    Consider an infinite sequence $C_0, C_1, C_2, \ldots$ of state classes.
    Let $C_i=(m_i,D_i)$.

    From \cite{boucheneb-FORMATS-17}, we know that $\costrincl$ is a wqo for the state classes of bounded (non parametric) cost TPNs. So for each integer parameter valuation $v$, and using a classical property of wqo we can extract a subsequence of $v(C_{0}), v(C_{1}), \ldots$ that is completely ordered by $\costrincl$. And since, we have a finite number of such parameter valuations, we can extract an infinite subsequence $C_{i_0}, C_{i_1}, \ldots$ such that for all integer parameter valuations $v$, $v(C_{i_0}) \costrincl v(C_{i_1}) \costrincl \cdots$.

    Let us consider two of those: $C_{i_r}$ and $C_{i_s}$, with $r<s$.

    Since $\IH(D_{i_s})$ has integer vertices, and for any integer parameter valuation, $v(C_{i_r}) \costrincl v(C_{i_s})$, which implies that $v(D_{i_s})\subseteq v(D_{i_r})$, then all the vertices of $D_{i_s}$ are also in $D_{i_r}$. Now assume that some extremal ray $\vec{r}$ of $D_{i_s}$ is not in $D_{i_r}$. Then starting from some vertex $\vec{x}$ of $D_{i_r}$, there must be some $\lambda\leq 0$ such that $\vec{x}+\lambda r\not\in D_{i_s}$ and the same holds for any $\lambda'\geq \lambda$ (by convexity). But since $r$ has rational coordinates for some value of $\lambda'$, $\lambda'r$ is an integer vector and so is $\vec{x}+\lambda'r$, which contradicts the fact that $v(D_{i_s})\subseteq v(D_{i_r})$, for all integer parameter valuations $v$, and in particular $\proj{(\vec{x}+\lambda'r)}{\Params}$. We can therefore conclude that $D_{i_r}\subseteq D_{i_s}$ and we now proceed to proving that $D_{i_s}$ is also ``cheaper'' than $D_{i_r}$.

    We use another property of the vertices of convex polyhedra: vertices of a convex polyhedron of dimension $n$ defined by $m$ inequalities $\sum_{k=1}^n a_{kl}x_k \leq b_l$, for $j\in [1..m]$ are solutions of a system of $n$ linearly independent equations $\sum_{k=1}^n a_{kl}x_k = b_l$, with $l$ in a subset of size $n$ of $[1..m]$.

    Now consider the polyhedron $D$ obtained from $\IH(D_{i_r})$, with its cost variable $c$, by adding one variable $c'$ constrained by the cost inequalities of $\IH(D_{i_s})$.
    Clearly, since $c$ and $c'$ are not constrained together, the vertices of $D$ are those of $\IH(D_{i_r})$, extended with the corresponding minimal and maximal values of $c'$, and symmetrically those of $\IH(D_{i_s})$, extended with the corresponding minimal and maximal values of $c'$. Since the inequalities constraining $c$ and $c'$ have integer coefficients, and $\IH(D_{i_s})$ and $\IH(D_{i_r})$ have integer vertices, $D$ also has integer vertices.

    For the $i$-th lower-bound inequality on $c$, and the $j$-th lower-bound inequality on $c'$, we define $E_{ij}$ as $D$ in which we transform both constraints into equalities. Clearly, from the property above, this does not add any new vertex, but it may remove some. Second, by construction, we have $\bigcup_{ij} E_{ij} = \{(\vec{\theta},\min_{(\vec{\theta},c,v)\in \IH(D_{i_r})} c, \min_{(\vec{\theta},c,v)\in \IH(D_{i_s})} c) | \vec{\theta}\in\proj{\IH(D_r)}{\Theta}\}$.
    If we minimize $c-c'$ over $E_i$, we know from the theory of linear programming that the minimum is obtained at a vertex of $E_{ij}$, and therefore, in particular, for an integer valuation $v$ of the parameters, and an integer vector $\vec{\theta}$ of $D_{i_r}$. Since we have $v(C_{i_r}) \costrincl v(C_{i_s})$, we then know that for these values of the theta variables and parameters, $c\leq c'$. This means that this holds for the whole of $E_{ij}$, and finally that $C_{i_r}\costrincl C_{i_s}$.
\end{proof}

\section{Case study}
\label{sec:casestudy}
We now consider a scheduling problem where some tasks include \emph{runnables}, a key concept of the AUTomotive Open System ARchitecture (AUTOSAR), the open standard for designing the architecture of vehicle software~\cite{Autosar-RTE}.
Runnables represent the functional view of the system and are executed by the runtime of the software component~\cite{naumann2009autosar}. For their execution they are mapped to tasks and a given runnable can be split across different tasks to introduce parallelism, for instance. In industrial practice, runnables that interact a lot are mapped to the same task, in particular when they perform functions with the same period.

\medskip
In this example, we consider 3~non-preemptive, periodic tasks T1, T2 and~T3, on which have already been mapped some runnables that interact together; we add another independent runnable whose code must be split between tasks T1 and T2:
\begin{itemize}
	\item the period of task T1 is 100~time units; T1 includes a ``fixed part'', independent from the new runnable and whose execution lasts 22\,t.u.;
	\item the period of T2 is 200\,t.u.; T2 also has a fixed part lasting 28\,t.u.;
	\item the period of T3 is 400\,t.u.; its execution lasts 11\,t.u.;
	\item the period of the runnable is 200\,t.u.; its execution lasts 76\,t.u.; parameter $a$ denotes the duration of the section that is executed in T1\footnote{Every 200\,t.u., since T1 is executed twice as often as T2, T1 is running during $(22+a)*2=44+2a$\,t.u. whereas T2 is running during $28+(76-2a)=104-2a$\,t.u.}.
\end{itemize}

The processing unit consists of 2~cores C0 and C1; T3 can only execute on C0 whereas both T1 and~T2 can execute on either core.
When both cores are idle, the cost is null; when only one core is busy, the cost is equal to 2/t.u.; when both cores are busy, the cost is equal to 3/t.u.
Any optimised strategy to divide the runnable over T1 and T2 and to allocate these tasks to C0 or C1 must therefore favour the cases where both cores are in the same state.

\begin{figure}[h]
    \centering
\includegraphics[width=13cm]{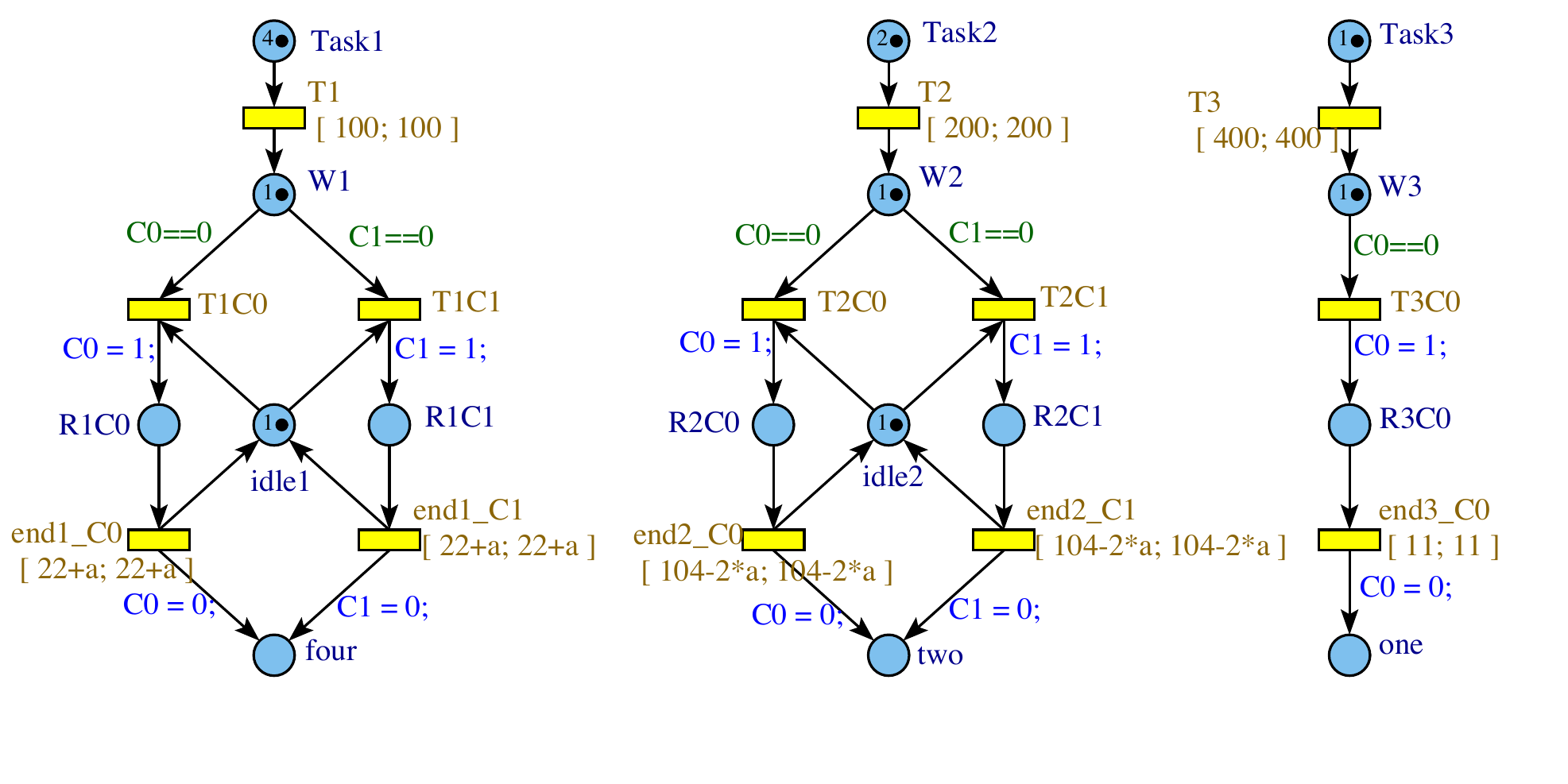}
\vspace{-1.2cm}
\caption{Offline non preemptive scheduling problem}
\label{fig:scheduling}
\end{figure}

Figure~\ref{fig:scheduling} presents the model for this problem\footnote{To ensure a correct access to the cores, we could have added one place for each core and some arcs on each task to capture and release them but the resulting net would have been quite unreadable.
Instead, we chose to add 2~integer variables C0 and C1 (both initialised to~0); a variable equal to~0 (resp.~1) obviously means the corresponding core is idle (resp. busy).}.
The associated cost function is:
{\small$2*(C0\neq C1) + 3*C0*C1 + 1000*\big(W1*(R1C0+R1C1) + W2*(R2C0+R2C1) + W3*R3C0\big)$}, where the name of a place (e.g. $R1C0$) represents its marking\footnote{The last term ensures that such cases where an instance of a task is activated while a previous one is running are heavily penalised.}.

We limit the study of the system to the first 400\,t.u., at the end of which T1 has been executed 4~times, T2 twice and T3 once. This will be the marking to reach.

We start by computing the optimal cost with the corresponding parameter values. This is done with Romeo using formula \texttt{mincost (four==4 and two==2 and one==1)}. We obtain minimum cost $466$ and $a \in[13,17]$.

We can run a consistency check with the following bounded cost reachability property: \texttt{EF four==4 and two==2 and one==1 and cost}$\leq$\texttt{466}. As expected this holds iff $a \in[13,17]$.

We then try to relax the constraint on the cost a bit and for instance we find that property \texttt{EF four==4 and two==2 and one==1 and cost}$\leq$\texttt{470} holds iff $a\in[12,18]$. Similarly, the cost can be made less or equal to $500$ iff $a\in [4,26]$.

Back to the optimal case, we set $a$ to~17; Romeo provides the following timed trace, in which the notation {\scriptsize T1@t1} means that transition T1 is fired at date t1:
{\scriptsize T1C0@61, T2C1@69, T1@100, end1\_C0@100, T1C0@100, end1\_C0@139, end2\_C1@139, T1@200, T2@200, T2C0@261, T1C1@261, T1@300, end1\_C1@300, T1C1@303, end2\_C0@331, T3C0@331, end3\_C0@342, end1\_C1@342}

From this trace, we obtain the Gantt chart in Figure~\ref{Gantt} (above).
Setting $a$ to~13 yields another timed trace, resulting in the Gantt chart in Figure~\ref{Gantt} (below).
In both cases, we can see that both cores are busy during 148\,t.u. (and for 11\,t.u., only one is idle), which confirms our analysis on the optimised strategy above.

\begin{figure}[h]
\begin{center}
\begin{tikzpicture}[scale=0.5]
\draw (-0.8,.5) node[left] {Task1};
\draw (0,0.25) node[left] {{\tiny C0}};
\draw (0,0.75) node[left] {{\tiny C1}};

\draw (-0.8,2) node[left] {Task2};
\draw (0,1.75) node[left] {{\tiny C0}};
\draw (0,2.25) node[left] {{\tiny C1}};

\draw (-0.8,3.32) node[left] {Task3};
\draw (0,3.25) node[left] {{\tiny C0}};

\fill[gray!75] (69/20,2) rectangle ++(70/20,.5);
\fill[gray!75] (261/20,1.5) rectangle ++(70/20,.5);

\fill[gray!75] (61/20,0) rectangle ++(39/20,.5);
\fill[gray!75] (100/20,0) rectangle ++(39/20,.5);
\fill[gray!75] (261/20,0.5) rectangle ++(39/20,.5);
\fill[gray!75] (303/20,0.5) rectangle ++(39/20,.5);

\fill[gray!75] (331/20,3) rectangle ++(11/20,.5);

\draw[dotted] (0,3.5) -- ++(20,0);
\draw[dotted] (0,2) -- ++(20,0);
\draw[dotted] (0,2.5) -- ++(20,0);
\draw[dotted] (0,0.5) -- ++(20,0);
\draw[dotted] (0,1) -- ++(20,0);

\draw[->,thick] (0,0) -- (20.5,0);
\draw[->,thick] (0,1.5) -- (20.5,1.5);
\draw[->,thick] (0,3) -- (20.5,3);
\draw (20.5,0) node[right] {$t$};

\foreach \x in {0,20}
{
\draw [->,thick] (\x,3) -- ++(0,0.8);
}

\foreach \x in {1,...,40}
{
  \draw[very thin] (\x/2,0) -- ++(0,1);
  \draw[very thin] (\x/2,1.5) -- ++(0,1);
  \draw[very thin] (\x/2,3) -- ++(0,0.5);
}
\foreach \x in {0,100,200,300,400}
{
\draw [->,thick] (\x/20,0) -- ++(0,1.3);
}
\foreach \x in {0,200,400}
{
\draw [->,thick] (\x/20,1.5) -- ++(0,1.3);
}
\foreach \x in {0,40,80,120,...,400}
{
  \draw (\x/20,0) node[below] {$\x$};
}
\foreach \x in {40,80,120,...,360}
{
  \draw[ thin] (\x/20,-0.15) -- ++(0,3.65);
}

\end{tikzpicture}

\begin{tikzpicture}[scale=0.5]
\draw (-0.8,.5) node[left] {Task1};
\draw (0,0.25) node[left] {{\tiny C0}};
\draw (0,0.75) node[left] {{\tiny C1}};

\draw (-0.8,2) node[left] {Task2};
\draw (0,1.75) node[left] {{\tiny C0}};
\draw (0,2.25) node[left] {{\tiny C1}};

\draw (-0.8,3.32) node[left] {Task3};
\draw (0,3.25) node[left] {{\tiny C0}};

\fill[gray!75] (68/20,2) rectangle ++(78/20,.5);
\fill[gray!75] (265/20,2) rectangle ++(78/20,.5);

\fill[gray!75] (65/20,0) rectangle ++(35/20,.5);
\fill[gray!75] (100/20,0) rectangle ++(35/20,.5);
\fill[gray!75] (265/20,0) rectangle ++(35/20,.5);
\fill[gray!75] (300/20,0) rectangle ++(35/20,.5);

\fill[gray!75] (135/20,3) rectangle ++(11/20,.5);

\draw[dotted] (0,3.5) -- ++(20,0);
\draw[dotted] (0,2) -- ++(20,0);
\draw[dotted] (0,2.5) -- ++(20,0);
\draw[dotted] (0,0.5) -- ++(20,0);
\draw[dotted] (0,1) -- ++(20,0);

\draw[->,thick] (0,0) -- (20.5,0);
\draw[->,thick] (0,1.5) -- (20.5,1.5);
\draw[->,thick] (0,3) -- (20.5,3);
\draw (20.5,0) node[right] {$t$};

\foreach \x in {0,20}
{
\draw [->,thick] (\x,3) -- ++(0,0.8);
}

\foreach \x in {1,...,40}
{
  \draw[very thin] (\x/2,0) -- ++(0,1);
  \draw[very thin] (\x/2,1.5) -- ++(0,1);
  \draw[very thin] (\x/2,3) -- ++(0,0.5);
}
\foreach \x in {0,100,200,300,400}
{
\draw [->,thick] (\x/20,0) -- ++(0,1.3);
}
\foreach \x in {0,200,400}
{
\draw [->,thick] (\x/20,1.5) -- ++(0,1.3);
}
\foreach \x in {0,40,80,120,...,400}
{
  \draw (\x/20,0) node[below] {$\x$};
}
\foreach \x in {40,80,120,...,360}
{
  \draw[ thin] (\x/20,-0.15) -- ++(0,3.65);
}

\end{tikzpicture}
\end{center}
\vspace{-0.6cm}
\caption{Gantt charts for $a=17$ (above) and $a=13$ (below)}
\label{Gantt}
\end{figure}
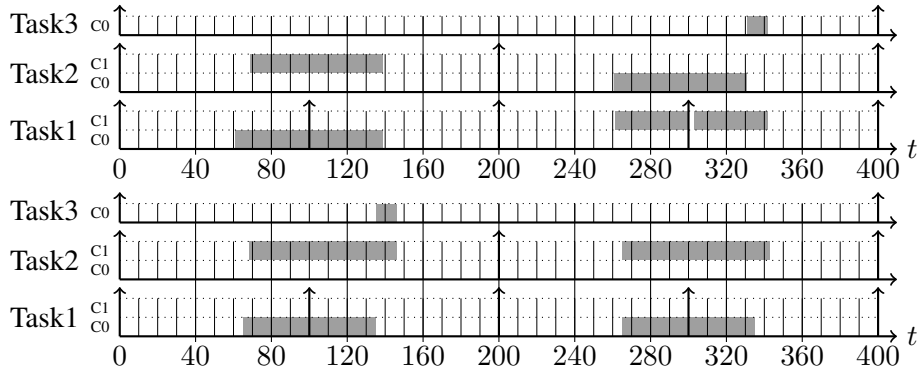

\section{Conclusion}
\label{sec:conclusion}

We have proposed a new Petri net-based formalism with parametric timing and cost features, thus merging two classical lines of work.
For this formalism, we define an existential problem and two synthesis problems for parametric reachability with cost constraints.
We prove that the existential problem is undecidable but we nonetheless give and prove symbolic semi-algorithms for the synthesis problems.
We finally propose variants of those synthesis semi-algorithms suitable for integer parameter valuations and prove their termination when those parameter valuations are bounded a priori, and with some other classical assumptions.
These symbolic algorithms avoid the explicit enumeration of all possible parameter valuations. They are implemented in our tool Romeo and we have demonstrated their use on a case-study addressing a scheduling problem, and inspired by the AUTOSAR standard.

Further work includes computing the optimal cost as a function of parameters and investigating the case of costs (discrete and rates) as parameters.



\end{document}